\documentclass[11pt,reqno]{amsart}
\usepackage{a4wide}

\usepackage{amssymb, amsbsy, amsfonts}
\usepackage{amsmath, amsthm, mathrsfs}
\usepackage{latexsym,float,color}

\allowdisplaybreaks[4]

\def\<#1>{\langle#1\rangle}
\let\set\mathbb
\def\vect#1{\mathbf{#1}}

\def\disp{\operatorname{disp}}

\newcommand{\ve}[1]{\textit{\textbf{#1}}}

\def\ev{\operatorname{ev}}

\def\expr{\operatorname{expr}}

\def\sigmaSE{$\Sigma^*$}
\def\rpisiSE{$R\Pi\Sigma^*$}

\def\seqK{{\mathcal S}(\KK)}

\def\VV{\set V}

\def\AA{\set A}

\def\KK{\set K}
\def\NN{\set N}
\def\ZZ{\set Z}
\def\GG{\set G}
\def\HH{\set H}

\def\EE{\set E}
\def\QQ{\set Q}

\def\AH{\text{$\mathcal H$}}

\def\shiftS{{\mathcal S}}
\def\Sv{\bar{S}}
\newcommand{\abs}[1]{\lvert#1\rvert}
\newcommand{\sign}[1]{\textnormal{sign}(#1)}
\newcommand{\ie}{i.e.,\ }

\newdimen\listablecorrection
\newtheorem{theorem}{Theorem}
\newtheorem{remark}{Remark}
\newtheorem{proposition}{Proposition}
\newtheorem{corollary}{Corollary}
\newtheorem{lemma}{Lemma}
\newtheorem{definition}{Definition}
\newtheorem{example}{Example}

\newcommand{\fct}[3]{#1:#2\to#3}

\newcommand{\dfield}[2]{(#1,#2)}

\newcommand{\const}[2]{\operatorname{const}(#1,#2)}

\title{Algebraic independence of sequences generated by (cyclotomic) harmonic sums}

\author{Jakob Ablinger}
\address{Research Institute for Symbolic Computation\\
J. Kepler University Linz\\
A-4040 Linz, Austria}
\email{Jakob.Ablinger@risc.jku.at}

\author{Carsten Schneider}
\address{Research Institute for Symbolic Computation\\
J. Kepler University Linz\\
A-4040 Linz, Austria}
\email{Carsten.Schneider@risc.jku.at}
\thanks{Supported by the Austrian Science Fund (FWF) grant SFB F50 (F5009-N15).}


\begin{document}

\begin{abstract}
An expression in terms of (cyclotomic) harmonic sums can be simplified by the quasi-shuffle algebra in terms of the so-called basis sums. By construction, these sums are algebraically independent within the quasi-shuffle algebra. In this article
we show that the basis sums can be represented within a tower of difference ring extensions where the constants remain unchanged. This property enables one to embed this difference ring for the (cyclotomic) harmonic sums into the ring of sequences. 
This construction implies that the sequences produced by the basis sums are algebraically independent over the rational sequences adjoined with the alternating sequence. 
\end{abstract}

\maketitle

\section{Introduction}

Special functions like the harmonic numbers and more generally indefinite nested sums defined over products play a dominant role in many research branches, like in combinatorics, number theory, and in particle physics. For concrete examples within these research areas in connection with symbolic summation see, e.g.,~\cite{Comb1,Comb2},~\cite{Number1,Number2} and~\cite{QCD1,QCD2}, respectively. In particular, these nested sums cover the class of d'Alembertian solutions~\cite{Abramov:94}, a sub-class of Liouvillian solutions~\cite{Singer:99}, of linear recurrence relations; for further details see~\cite{Petkov:2013}.\\
Numerous properties of such sum classes, like the harmonic sums~\cite{Bluemlein1999,Vermaseren1998}, cyclotomic harmonic sums~\cite{Ablinger2011}, generalized harmonic sums~\cite{Moch2002,SSum} or binomial sums~\cite{Binom1,Binom2,Binom3,Binom4a} have been explored.
In particular, the connection of the nested sums to nested integrals (i.e., to multiple polylogarithms and generalizations of them) via the (inverse) Mellin transform~\cite{Remiddi2000}, the analytic continuation~\cite{Bluemlein2000,Bluemlein2005} of nested sums or the calculation of asymptotic expansions of such sums~\cite{Lille,Bluemlein2009,Bluemlein2009a} has been worked out. For further details and generalizations of these results we refer to~\cite{Ablinger2011,SSum,Binom4a}. The underlying algorithms are implemented in the Mathematica package \texttt{HarmonicSums}~\cite{AblingerThesis:12,HarmonicSums}. 

Among all these algorithmic constructions, a key technology is the elimination of algebraic dependencies of the arising nested sums within a given expression to gain compact representations. Here the Mathematica package \texttt{Sigma}~\cite{Sigma1,Sigma2} provides strong tools that can simplify, among many other features, an expression in terms of indefinite nested product-sums to an expression in terms of such sums that are all algebraically independent in the analysis sense~\cite{AlgebraicDF,AlgebraicDF2,DR1,DR2,Schneider:17}. This means that the sequences with entries from a field $\KK$, that are produced by the reduced sums, are algebraically independent. In order to accomplish this task, the arising sums and products are represented in a difference ring, i.e., the sum objects are represented in a ring $\AA$ and the shift behaviour of the sums is modelled by a ring automorphism $\fct{\sigma}{\AA}{\AA}$. More precisely, the sums and products are represented in an \rpisiSE-extension\footnote{For the corresponding difference field theory see~\cite{Karr:81,Karr:85}.}~\cite{DR1,DR2} with the distinguished property that the set of constants is precisely the field $\KK$, i.e.,
$$\{c\in\AA|\,\sigma(c)=c\}=\KK.$$
Exactly this property enables one to embed the ring $\AA$ into the ring of sequences. 
This technology has been used to show in~\cite{DR1} that the sequences of the generalized harmonic numbers are algebraically independent over the rational sequences. In particular,
fast summation algorithms~\cite{FastAlgorithm1,FastAlgorithm2,FastAlgorithm3} in the setting of difference rings and fields support this construction algorithmically
and expressions with up to several hundred algebraically independent sums can be generated automatically. 
However, recently we were faced with QCD calculations~\cite{QCD3} with expressions of about 1GB and more than 20000 sums. At this level, the difference ring algorithms failed to eliminate all algebraic relations in a reasonable amount of time.

In order to perform such large scale calculations, another key property of certain classes of indefinite nested sums can be utilized:
they obey quasi-shuffle algebras~\cite{Hoffman1992,Hoffman1997,Hoffman}. This enables one to rewrite any polynomial expression in terms of indefinite nested sums as a linear combination of indefinite nested sums. As worked out in~\cite{Bluemlein2004} and continued in~\cite{Ablinger2011,AblingerThesis:12,SSum}, this feature can be used to hunt for algebraic relations among the occurring indefinite nested sums and to express the compact result in terms of the so-called basis sums which cannot be eliminated further by the quasi-shuffle algebra. Using the \texttt{HarmonicSums} package expressions as mentioned above could be reduced to several MB in terms of about several thousand basis sums; for details see~\cite{QCD3}. 
Summarizing, using the property of the underlying quasi-shuffle algebra one obtains dramatic compactifications within the demanding calculations in particle physics.

A natural question is if the obtained sums induced by the quasi-shuffle algebra are also algebraically independent in the sense of analysis, i.e., if the sequences produced by the nested sums are algebraically independent. 
A special variant for non-alternating harmonic sums has been accomplished in~\cite{Lille} using the knowledge of certain integral representations. In the following we will focus on the general case for the harmonic sums~\cite{Bluemlein1999,Vermaseren1998}
\begin{eqnarray}\label{Equ:HarmonicSumsIntro}
	\sum_{n\geq i_1 \geq i_2 \geq \cdots \geq i_k \geq 1} \frac{\sign{c_1}^{i_1}}{i_1^{\abs {c_1}}}\cdots
	\frac{\sign{c_k}^{i_k}}{i_k^{\abs {c_k}}}	
\end{eqnarray}
with non-negative integers $n$ and non-zero integers $c_i$ $(1 \leq
i \leq k)$ 
and for their cyclotomic versions~\cite{Ablinger2011}: for $\KK$ being a field containing the rational numbers the summand is of the kind\footnote{$\NN$ denotes the positive integers and $\NN_0=\NN\cup\{0\}$.} 
\begin{equation}
\frac{z_j^{i_j}}{(a_j\,i_j + b_j)^{c_j}},~~~~a_j,b_j,c_j \in \mathbb{N},~z_j\in\KK\setminus\{0\}
\end{equation}
and $i_j$ denotes the summation variable.

In this article we will consider the so-called basis sums induced by the quasi-shuffle algebra. This means we consider a particular chosen set of nested sums that generate all other nested sums and that do not possess any further relations using the quasi-shuffle algebra operation. Our main result is that these basis sums are also algebraically independent as sequences. More precisely, consider the ring of sequences which is defined by the
set of sequences $\KK^{\NN_0}=\{\langle a_n\rangle_{n\geq0}|a_n\in\KK\}$ equipped with component-wise addition and multiplication
where two sequences are identified as equal if they differ only by finitely many entries. Then we will show that the basis sums evaluated to such elements of the ring of sequences are algebraically independent: they are algebraically independent over the sub-ring of sequences that is generated by all rational functions from $\KK(n)$ and $(-1)^{n}$.
We will derive this result by showing that the basis sums generate an \rpisiSE-extension in the difference ring sense. This means that the basis sums generate a polynomial ring equipped with a shift operator such that the set of constants is precisely $\KK$. Based on this particularly nice  structure it will follow by difference ring theory~\cite{Singer:97,Schneider:17} that this difference ring can be embedded by an injective difference ring homomorphism into the ring of sequences. In other words, the algebraic properties of the polynomial ring (in particular, the algebraic independence of variables of the polynomial ring, which are precisely the basis sums) carry over into the setting of sequences.

The outline of the article is as follows. In Section~\ref{Sec:GeneralOutLine} we will set up the general framework for (cyclotomic) harmonic sums. In Section~\ref{Sec:BasicDR} we will present basic constructions to represent (cyclotomic) harmonic sums in a difference ring. In Section~\ref{Sec:QuasiShuffle} we will introduce the quasi-shuffle algebra for (cyclotomic) harmonic sums and will work out various properties that link the quasi-shuffle algebra with our difference ring construction. In Section~\ref{Sec:ReduceDR} we define the reduced difference ring for (cyclotomic) harmonic sums in which all algebraic relations are eliminated that are induced by the quasi-shuffle algebra. We will provide new structural results obtained by the difference ring theory of \rpisiSE-extensions in Section~\ref{Sec:pisiRings}.  
In Section~\ref{Sec:MainResult} we will combine all these results and will show that our reduced difference ring is built by a tower of \rpisiSE-extensions. As a consequence we can conclude that this ring can be embedded into the ring of sequences. A conclusion is given in Section~\ref{Sec:Conclusion}.

\section{A general framework for cyclotomic harmonic sums}\label{Sec:GeneralOutLine}

Throughout this article we assume that $\KK$ is a field containing $\QQ$ as a subfield. In particular, we assume that there is a linear ordering $<$ on $\KK$. 
For a set $B$, $B^*$ denotes the set of all finite words over $B$ (including the empty word), i.e., 
$$B^*=\{b_1,\dots,b_k|k\geq 0\text{ and } b_i\in B\}.$$
Furthermore, we define the alphabet 
$$A:=\left\{(a,b,c,z)|a,c\in\NN, b\in\NN_0, z\in\KK\setminus\{0\}\text{ with }b<a\text{ and } \gcd(a,b)=1 \right\}$$
as a totally ordered, graded set. More precisely, the degree
of $(a,b,c,d)\in A$ is denoted by $\abs{(a,b,c,d)}:=c$. This establishes the grading $A_i=\{a\in A|\,|a|=i\}$.
Moreover, we define the linear order $<$ on $A$ in the following way:
\begin{align*}
 (a_1,b_1,c_1,z_1)<(a_2,b_2,c_2,z_2) &\quad\text{if } c_1<c_2\\
 (a_1,b_1,c,z_1)<(a_2,b_2,c,z_2) &\quad\text{if } a_1<a_2\\
 (a,b_1,c,z_1)<(a,b_2,c,z_2) &\quad\text{if } b_1<b_2\\
 (a,b,c,z_1)<(a,b,c,z_2) &\quad\text{if } z_1<z_2.
\end{align*}
Furthermore, we define the function
\begin{equation}\label{Equ:LambdaFu}
\begin{split}
\lambda:A\times\NN \to \KK\\
\lambda((a,b,c,z),i)\mapsto \frac{z^i}{(a i+b)^c}.
\end{split}
\end{equation}
Note that $\lambda((1,0,c_1,z_1),i)\lambda((1,0,c_2,z_2),i)=\lambda((1,0,c_1+c_2,z_1z_2),i)$. For arbitrary letters in $A$ the connection is more complicated but there is always a relation of the form
\begin{equation}\label{Equ:LambdaClosure}
\lambda((a_1,b_1,c_1,z_1),i)\lambda((a_2,b_2,c_2,z_2),i)=\sum_{j=1}^kr_j\lambda((e_j,f_j,g_j,h_j),i)
\end{equation}
with $r_j\in \QQ,$ $(e_j,f_j,g_j,h_j)\in A$ and $g_j\leq c_1+c_2$ see, e.g., \cite{Ablinger2011,AblingerThesis:12}.
For $n \in \NN_0$, $k\in\NN$, $a_i\in A$ with $1\leq i\leq k$ we define nested sums (compare \cite{Ablinger2011,AblingerThesis:12})
\begin{eqnarray}
	S_{a_1}(n)&=& \sum_{i=1}^n \lambda(a_1,i)\nonumber\\
	S_{a_1,\ldots ,a_k}(n)&=& \sum_{i=1}^n \lambda(a_1,i)S_{a_2,\ldots ,a_k}(i). 
	\label{defHsum}
\end{eqnarray}
Moreover, we define the weight function $w$ on these nested sums: $w(S_{a_1,a_2,a_3,\ldots ,a_k}(n))=\abs{a_1}+\cdots+\abs{a_k}$ and extend it to monomials such that
the weight of a product of nested sums is the sum of the weights of the individual sums, i.e.,
$$
w(S_{\ve {a}_1}(n)S_{\ve {a}_2}(n)\cdots S_{\ve {a}_k}(n))=w(S_{\ve {a}_1}(n))+w(S_{\ve {a}_2}(n))+\cdots +w(S_{\ve {a}_k}(n)).
$$
Instead of $S_{a_1,a_2,\ldots ,a_k}(n)$ we will also write
$S_{a_1a_2\cdots a_k}(n)$ or $S_{\vect{a}}(n)$ with $\vect{a}=a_1a_2\dots a_k\in A^*$. 

A product of two nested sums with the same upper summation limit can be written in terms of single nested sums: for $n\in \NN_0,$
\begin{eqnarray}
	S_{a_1,\ldots ,a_k}(n)S_{b_1,\ldots ,b_l}(n)&=&
	\sum_{i=1}^n \lambda(a_1,i) S_{a_2,\ldots ,a_k}(i)\,S_{b_1,\ldots ,b_l}(i) \nonumber\\
	&&+\sum_{i=1}^n \lambda(b_1,i)S_{a_1,\ldots ,a_k}(i)\,S_{b_2,\ldots ,b_l}(i) \nonumber\\
	&&-\sum_{i=1}^n \lambda(a_1,i)\lambda(b_1,i)S_{a_2,\ldots ,a_k}(i)\,S_{b_2,\ldots ,b_l}(i).
\label{hsumproduct}
\end{eqnarray}
Note that the product of the two sums within the summands of the right side can be expanded further by using again this product formula. Applying this reduction recursively will lead to a linear combination of sums $S_{b_1,\dots,b_r}(n)$ with $b_i\in A$. In particular, the maximum of all the weights of the derived sums is precisely the weight of the left hand side expression. 

We can consider different subsets of~$A:$ 
\begin{enumerate}
\item If we consider only letters of the form $(1,0,c,1)$ with $c\in \NN$, i.e., we restrict to
$$A_h:=A\cap(\{1\}\times\{0\}\times\NN\times\{1\}),$$
then we are dealing with harmonic sums see, e.g.,\cite{Bluemlein2004,Vermaseren1998}.
\item If we consider only letters of the form $(1,0,c,\pm 1)$ with $c\in \NN$,
i.e., we restrict to
$$A_a:=A\cap(\{1\}\times\{0\}\times\NN\times\{1,-1\}),$$
then we are dealing with alternating harmonic sums see, e.g.,\cite{Bluemlein2004,Vermaseren1998}.
\item Let $M\subset\NN$ be a finite subset of $\NN$. If we consider only letters of the form $(a,b,c,\pm 1)$ with $a\in M;c\in\NN;b\in\NN_0,b\leq a$, $\gcd(a,b)~=~1$,
i.e., we restrict to
$$A_{c(M)}:=A\cap(M\times\NN_0\times\NN\times\{1,-1\}),$$
then we are dealing with cyclotomic harmonic sums see, e.g.,\cite{Ablinger2011,AblingerThesis:12}.
\item If we consider only letters of the form $(a,b,c,\pm 1)$ with $a,c\in\NN;b\in\NN_0,b\leq a$, $\gcd(a,b)~=~1$,
i.e., we restrict to
$$A_c:=A\cap(\NN\times\NN_0\times\NN\times\{1,-1\}),$$
then we are dealing with the full set of cyclotomic harmonic sums see, e.g.,\cite{Ablinger2011,AblingerThesis:12}.
\end{enumerate}
Note that for every finite subset $M$ of $\NN$ we have 
$$A_h\subset A_a\subset A_{c(M)}\subset A_c.$$
Throughout this article we will assume that 
\begin{equation}\label{Equ:ChooseH}
\AH\in\{A_h,A_a,A_{c(M)}\}
\end{equation}
holds. In particular, we call a sum $S_{a_1a_2\dots a_k}(n)$ with $a_i\in\AH$ also \AH-sum.

\section{A basic difference ring construction for the expression of \AH-sums}\label{Sec:BasicDR}

In the following we will define a difference ring in which we will represent the expressions of \AH-sums.

\begin{definition}
An expression of \AH-sums in $n$ over a field $\KK$ is built by
\begin{enumerate}
 \item rational expressions in $n$ with coefficients from $\KK$, i.e., elements from the rational function field $\KK(n)$,
 \item $(-1)^n$ that occurs in the numerator,
 \item the \AH-sums that occur as polynomial expressions in the numerator.
\end{enumerate}
\end{definition}
If $f$ is such an expression we use for $\lambda\in\KK(n)$  the shortcut
$$f(\lambda):=f|_{n\to\lambda}.$$
Sometimes we also use the notation $f(n)$ to indicate that the expression depends on a symbolic variable $n$. We say that an expression $e(n)$ of \AH-sums has no pole for all $n\in\NN_0$ with $n\geq\lambda$ for some $\lambda\in\NN_0$, if the rational functions occurring in $e(n)$ do not introduce poles at any evaluation $n\to\nu$ for $\nu\in\NN_0$ with $\nu\geq\lambda$. If this is the case, one can perform the evaluation $e(\nu)$ for all $\nu\in\NN$ with $\nu\geq\lambda$.
For a more rigorous definition of indefinite nested product-sum expressions (containing as special case the \AH-sums) in terms of term algebras, we refer to~\cite{AlgebraicDF2} which is inspired by~\cite{PauleNemes:97}. 

These expressions will be represented in a commutative ring $\AA$ and the shift operator acting on the expressions in terms of \AH-sums will be rephrased by a ring automorphism $\fct{\sigma}{\AA}{\AA}$.
Such a tuple $\dfield{\AA}{\sigma}$ of a ring $\AA$ equipped with a ring automorphism is also called difference ring; if $\AA$ is a field, $\dfield{\AA}{\sigma}$ is also called a difference field. In such a difference ring we call $c\in\AA$ a constant if $\sigma(c)=c$ and denote the set of constants by
$$\const{\AA}{\sigma}=\{\sigma(c)=c|\,c\in\AA\}.$$
In general, $\const{\AA}{\sigma}$ is a subring of $\dfield{\AA}{\sigma}$. But in most applications we take care that $\const{\AA}{\sigma}$ itself forms a field.

Our construction will be accomplished step by step. Namely, suppose that we are given already a difference ring $\dfield{\AA^{\AH}_d}{\sigma}$ in which we succeeded in representing parts of our \AH-sums. In order to enrich this construction, we will extend the ring from $\AA^{\AH}_d$ to $\AA^{\AH}_{d+1}$ and will extend the ring automorphism $\sigma$ to a ring automorphism $\fct{\sigma'}{\AA^{\AH}_{d+1}}{\AA^{\AH}_{d+1}}$, i.e., for any $f\in\AA^{\AH}_d$ we have that $\sigma'(f)=\sigma(f)$. We say that such a difference ring $\dfield{\AA^{\AH}_{d+1}}{\sigma'}$ is a difference ring extension of $\dfield{\AA^{\AH}_d}{\sigma}$; in short, we also write 
$\dfield{\AA^{\AH}_{d}}{\sigma}\leq\dfield{\AA^{\AH}_{d+1}}{\sigma'}$. Since $\sigma$ and $\sigma'$ agree on $\AA^{\AH}_d$, we usually do not distinguish anymore between $\sigma$ and $\sigma'$.

We start with the rational function field $\KK(n)$ and define the field/ring automorphism $\fct{\sigma}{\KK(n)}{\KK(n)}$ with
$$\sigma(f)=f|_{n\to n+1}.$$
It is easy to verify that $\const{\KK(n)}{\sigma}=\const{\KK}{\sigma}=\KK$.
So far, we can model rational expressions in $n$ in the field $\KK(n)$ and can shift these elements with $\sigma$.

Next, we want to model the object $(-1)^n$ with the relations $((-1)^n)^2=1$ and $(-1)^{n+1}=-(-1)^n$. Therefore we take the ring $\KK(n)[x]$ subject to the relation $x^2=1$. Then one can verify that there is a unique difference ring extension $\dfield{\KK(n)[x]}{\sigma}$ of $\dfield{\KK(n)}{\sigma}$ with $\sigma(x)=-x$. In particular, we have that
$$\const{\KK(n)[x]}{\sigma}=\const{\KK(n)}{\sigma}=\KK.$$
Precisely, this difference ring $\dfield{\KK(n)[x]}{\sigma}$ enables one to represent all rational expressions in $n$ together with objects $(-1)^n$ that are rephrased by $x$. 

Before we can continue with our construction for \AH-sums, we observe the following easy, but important fact. 
\begin{lemma}\label{Lemma:SumExtension}
Let $\dfield{\AA}{\sigma}$ be a difference ring and let $\AA[t]$ be a polynomial ring, i.e., $t$ is transcendental over $\AA$, and let $\beta\in\AA$. Then there is a unique difference ring extension $\dfield{\AA[t]}{\sigma}$ of $\dfield{\AA}{\sigma}$ with $\sigma(t)=t+\beta$.
\end{lemma}
We will use this lemma iteratively in order to adjoin all \AH-sums to the difference ring $\dfield{\KK(n)[x]}{\sigma}$. This construction is done inductively on the weight of the sums. It is useful to define the following function (compare~\eqref{Equ:LambdaFu}):
\begin{eqnarray*}
\bar{\lambda}:\AH \to \KK(n)[x]\\
\bar{\lambda}((a,b,c,z))\mapsto \frac{x^{\frac{z-1}2}}{(a n+b)^c}.
\end{eqnarray*}

The base case is the already constructed difference ring $\dfield{\AA^{\AH}_0}{\sigma}$ with $\AA^{\AH}_0=\KK[x](n)$. Now suppose that we constructed the difference ring $\dfield{\AA^{\AH}_{d-1}}{\sigma}$ for all \AH-sums of weight $<d$. Then we will construct a difference ring extension $\dfield{\AA^{\AH}_{d}}{\sigma}$ which covers precisely the \AH-sums of weight $\leq d$. Consider all \AH-sums with weight $d\in\NN,$ say
$$S^{(d)}_1(n),\dots,S^{(d)}_{n_d}(n).$$
To these sums we attach the variables
\begin{equation}\label{Equ:WeightVariables}
s^{(d)}_1,\dots,s^{(d)}_{n_d} 
\end{equation}
of weight $d$, respectively. Now we define the polynomial ring 
$\AA^{\AH}_{d}=\AA^{\AH}_{d-1}[s^{(d)}_1,\dots,s^{(d)}_{n_d}]$. 
To this end, we extend $\sigma$ from $\AA^{\AH}_{d-1}$ to $\AA^{\AH}_d.$
Suppose that $s^{(d)}_i$ models the \AH-sums $S^{(d)}_i=S_{a_{i1},\dots,a_{ir_i}}$ with $a_{i1}+\dots+a_{ir}=d$. 
Note that
$$S^{(d)}_i(n+1)=S_{a_{i1},\dots,a_{ir_i}}(n+1)=S_{a_{i1},\dots,a_{ir_i}}(n)+\lambda(a_{i1},n+1)S_{a_{i2},\dots,a_{ir_i}}(n+1)$$
and
$$ S^{(d)}_i(n-1)=S_{a_{i1},\dots,a_{ir_i}}(n-1)=S_{a_{i1},\dots,a_{ir_i}}(n)-\lambda(a_{i1},n)S_{a_{i2},\dots,a_{ir_i}}(n)$$
where $S_{a_{i2},\dots,a_{ir_i}}$ is a \AH-sum of weight $d-\abs{a_{i1}}$. Let
$s^{(d-\abs{a_{i1}})}_i\in\AA^{\AH}_{d-\abs{a_{i1}}}$ which models this sum.
Therefore we extend $\sigma$ from $\AA^{\AH}_{d-1}$ to $\AA^{\AH}_d$ subject to the relations
\begin{equation}\label{Equ:ForwardShift}
\sigma(s^{(d)}_i)=s^{(d)}_i+\sigma(\bar{\lambda}(a_{i1}))\,\sigma(s^{(d-\abs{a_{i1}})}_i)
\end{equation}
by applying Lemma~\ref{Lemma:SumExtension} iteratively. This means, we first adjoin $s^{(d)}_1$ to $\AA^{\AH}_{d-1}$ and extend the automorphism with~\eqref{Equ:ForwardShift} for $i=1$, then we adjoin to this ring the variable $s^{(d)}_2$ and extend the automorphism with~\eqref{Equ:ForwardShift} for $i=2$, etc.  
We remark that this construction implies that 
\begin{equation}\label{Equ:BackwardShift}
\sigma^{-1}(s^{(d)}_i)=s^{(d)}_i-\bar{\lambda}(a_{i1})s^{(d-\abs{a_{i1}})}_i.
\end{equation}
By construction $\dfield{\AA^{\AH}_d}{\sigma}$ is a difference ring extension of $\dfield{\AA^{\AH}_{d-1}}{\sigma}$. 

Finally, we define the polynomial ring
$$\AA^{\AH}:=\KK(n)[x][s^{(1)}_1,\dots,s^{(1)}_{n_1}][s^{(2)}_1,\dots,s^{(2)}_{n_2}]\dots$$
with infinitely many variables, which represents all \AH-sums. In particular, we define the ring automorphism $\fct{\sigma'}{\AA^{\AH}}{\AA^{\AH}}$ as follows. For any $f\in\AA^{\AH}$,
we can choose a $d\in\NN_0$ such that $f\in\AA^{\AH}_{d}$. This defines\footnote{Note that any other choice $d'$ with $f\in\AA^{\AH}_{d'}$ will deliver the same evaluation.} $\sigma'(f):=\sigma(f)$ with $\fct{\sigma}{\AA^{\AH}_{d}}{\AA^{\AH}_{d}}$. By construction, $\sigma'|_{\AA^{\AH}_d}=\sigma$ where $\sigma$ is the automorphism of $\dfield{\AA^{\AH}_d}{\sigma}$. It is easy to see that $\dfield{\AA^{\AH}}{\sigma'}$ is a difference ring and that it is a difference ring extension of $\dfield{\AA^{\AH}_d}{\sigma}$. Again we do not distinguish anymore between $\sigma$ and $\sigma'$.
To sum up, we get the chain of difference ring extensions
\begin{equation}\label{Equ:NaiveChainDR}
\dfield{\KK(n)[x]}{\sigma}=\dfield{\AA^{\AH}_0}{\sigma}\leq\dfield{\AA^{\AH}_1}{\sigma}\leq\dfield{\AA^{\AH}_2}{\sigma}\leq\dots\leq\dfield{\AA^{\AH}}{\sigma}.
\end{equation}

For convenience, we will also write $\Sv_{a_1,\dots,a_k}$ for the variable $s^{(d)}_j$. In this way, we may write e.g.,
\begin{align*}
\sigma(\Sv_{a_{1},\dots,a_{ir_i}})&=\Sv_{a_{i1},\dots,a_{ir_i}}+\sigma(\bar{\lambda}(a_{i1}))\,\sigma(\Sv_{a_{i2},\dots,a_{ir_i}})\\
\sigma^{-1}(\Sv_{a_{i1},\dots,a_{ir_i}})&=\Sv_{a_{i1},\dots,a_{ir_i}}-\bar{\lambda}(a_{i1})\,\Sv_{a_{i2},\dots,a_{ir_i}}
\end{align*}
instead of~\eqref{Equ:ForwardShift} and~\eqref{Equ:BackwardShift}, respectively.

\medskip

To give a r\'{e}sum\'{e}, we can express every expression of \AH-sums over $\KK$ in $\dfield{\AA^{\AH}}{\sigma}$.
Conversely, if we are given a ring element $f\in\AA^{\AH}$, we denote by $\expr(f)$ the expression that is obtained when all occurrences of $x$ are replaced by $(-1)^n$ and all variables $s^{(d)}_i$ are replaced by the attached \AH-sums with upper summation range $n$. This will lead to an expression of \AH-sums in $n$ over $\KK$.  In this way, we can jump between the function and difference ring worlds. 

Now let $f,g\in\AA^{\AH}$. Then 
$$\expr(f+g)(n)=\expr(f)(n)+\expr(g)(n)\text{ and }\expr(f\,g)(n)=\expr(f)(n)\expr(g)(n);$$
recall that for an expression $e$ of \AH-sums, $e(n)$ is used to emphasize the dependence on the symbolic variable $n$.
If $\expr(f)(n)$ and $\expr(g)(n)$ have no poles for all $n\geq\delta$ for some $\delta\in\NN_0$, then it follows that 
$$\expr(f+g)(\lambda)=\expr(f)(\lambda)+\expr(g)(\lambda)\text{ and }\expr(f\,g)(\lambda)=\expr(f)(\lambda)\expr(g)(\lambda)$$ 
for all $\lambda\in\NN_0$ with $\lambda\geq\delta$.
Moreover observe that we model the shift-behaviour accordingly: For any $k\in\NN$ and any $\lambda\geq\delta$ we have that
\begin{align}
\expr(\sigma^k(f))(\lambda)&=\expr(f)(\lambda+k)\label{Equ:ExprForward}\\
\intertext{and for any $k\in\NN$ and any $\lambda\geq\delta+k$ we have that}
\expr(\sigma^{-k}(f))(\lambda)&=\expr(f)(\lambda-k).\label{Equ:ExprBackward}
\end{align}

The main goal of this article is to construct a difference ring, which represents all \AH-expressions and that can be embedded into the ring of sequences. 
As indicated already in the introduction, we will rely on the fact that the constants are precisely the elements $\KK$. The following example shows immediately, that $\dfield{\AA^{\AH}}{\sigma}$ is a too naive construction.

\begin{example}\label{Exp:AlgRelation}
Take 
$$f:=4 \Sv_{1,3}+6 \Sv_{2,2}+4 \Sv_{3,1}-12 \Sv_{1,1,2}-12 \Sv_{1,2,1}-12 \Sv_{2,1,1}+24 \Sv_{1,1,1,1}-\Sv_4-\Sv_1{}^4\in\AA^{\AH}.$$
Then one can easily verify that 
$\sigma(f)=f.$
Even more, we get that
$\expr(f)(\lambda)=0$
for all $\lambda\in\NN_0$, i.e., there are algebraic relations among these sums. 
\end{example}

\section{Quasi-shuffle algebras and the linearization operator}\label{Sec:QuasiShuffle}

In order to eliminate such relations as given in Example~\ref{Exp:AlgRelation},
we will equip the difference ring construction $\dfield{\AA^{\AH}}{\sigma}$ with the underlying quasi-shuffle algebra.

\begin{definition}[Non-commutative Polynomial Algebra]
\label{noncomalg}
Let $G$ be a totally ordered, graded set. The degree of $a\in G$ is denoted by $\abs{a}.$
Let $G^*$ denote the free monoid over $G$, i.e.,
		$$G^*=\left\{a_1 \cdots a_n | a_i \in G, n\geq 1\right\}\cup\left\{\epsilon \right\}.$$
We extend the degree function to $G^*$ by $\abs{a_1 a_2 \cdots a_n}=\abs{a_1}+\abs{a_2}+\cdots+\abs{a_n}$ for $a_i \in G$ and $\abs{\epsilon}=0.$
Let $R\supseteq \mathbb{Q}$ be a commutative ring. The set of non-commutative polynomials over $R$ is defined as
		$$
			R\left\langle G\right\rangle:=\left\{\left.\sum_{\ve w\in G^*} r_\ve w \ve w\right|r_\ve w \in R, r_\ve w=0 \textnormal{ for almost all } \ve w \right\}.
		$$
Addition in $R\left\langle G\right\rangle$ is defined component wise and multiplication is defined by
		$$
			\sum_\ve w a_\ve w \ve w \cdot \sum_\ve w b_\ve w \ve w :=  \sum_\ve w (\sum_{uv=\ve w} a_u b_v) \ve w. 
		$$
\end{definition}

We define a new multiplication $*$ on $R\left\langle G\right\rangle$ which is a generalisation of the \textit{shuffle product}, by requiring that $*$ distributes with the addition. 
We will see that this product can be used to describe properties of \AH-sums; compare~\cite{Hoffman1992,Hoffman1997,Hoffman}.
\begin{definition}[Quasi-shuffle product]
\label{quasishuffprodef}
$*: R\left\langle G\right\rangle \times R\left\langle G\right\rangle \longrightarrow R\left\langle G\right\rangle$ is called quasi-shuffle product, if it distributes with the addition and 
\begin{eqnarray}
\epsilon *\ve w &=& \ve w* \epsilon = \ve w, \ \textnormal{for all}\ w \in G^*,\nonumber\\
a\ve u*b\ve v &=& a(\ve u*\ve v)+b(\ve u*\ve v)-[a,b](\ve u*\ve v), \ \textnormal{for all}\ a,b \in G;\ve u,\ve v \in G^*,\label{eq:quasipro}
\end{eqnarray}
where $[\cdot,\cdot]:\overline{G}\times \overline{G}\rightarrow \overline{G}$, $(\overline{G} = G \cup \left\{0\right\})$ is a function satisfying
\begin{eqnarray}
&\textnormal{S}0.& [a,0]=0 \ \textnormal{for all } a \in \overline{G}; \nonumber\\
&\textnormal{S}1.& [a,b]=[b,a] \ \textnormal{for all } a,b \in \overline{G}; \nonumber\\
&\textnormal{S}2.& [[a,b],c]=[a,[b,c]] \ \textnormal{for all } a, b, c \in \overline{G}; \textnormal{ and} \nonumber\\
&\textnormal{S}3.& \textnormal{Either } [a,b]=0 \textnormal{ or } \abs{[a,b]}=\abs a +\abs b \ \textnormal{for all } a, b\in \overline{G}.\nonumber
\end{eqnarray}
\end{definition}


We specialize the quasi-shuffle algebra from Definition \ref{quasishuffprodef} in order to model the \AH-sums accordingly. We consider the alphabet $G=\AH$
and define the degree of a letter $\ve{a}$ by $\abs{\ve{a}}=w(\ve{a})$. Finally, we define 
\begin{equation}\label{Equ:CombineLambda}
[a,b]=\bar{\lambda}(a)\bar{\lambda}(b)
\end{equation}
and $[a,0]=0$ for all $a,b \in\AH$. This function obviously fulfils (S0)-(S3).
In other words, if we take our commutative ring $R=\KK(n)[x]$, then $\KK(n)[x]\langle\AH\rangle$ forms a quasi-shuffle algebra.

Let $\ve{a}_1,\dots,\ve{a}_r \in\AH^*$. By using the expansion of~\eqref{eq:quasipro}, we can write 
\begin{equation}\label{Equ:UniqueExpand}
\ve{a}_1*\dots *\ve{a}_k=c_1\ve{d}_1+\cdots+c_m\ve{d}_m
\end{equation}
for some uniquely determined $\ve{d}_i\in \AH^*, c_i\in\KK^*$ (compare \cite{Ablinger2011}). In particular, we have that 
\begin{equation}\label{Equ:DegreePreserving}
|\ve{a}_1|+\dots+|\ve{a}_k|=\max_{i=1,\ldots,m}|\ve{d}_i|.
\end{equation}


This linearization will be carried over to $\AA^{\AH}$. 
Consider the $\KK(n)[x]$-module 
\begin{equation}\label{Equ:LinearV}
\VV:=\KK(n)[x]s^{(1)}_1\oplus\dots\oplus\KK(n)[x]s^{(1)}_{n_1}\oplus\KK(n)[x]s^{(2)}_1\oplus\dots\oplus\KK(n)[x]s^{(2)}_{n_2}\oplus\dots.
\end{equation}
Now we are in the position to define the linearization function $\fct{L}{\AA^{\AH}}{\VV}$ as follows.
For $\ve{a}_1,\dots,\ve{a}_k\in\AH$, we take the $c_i\in\KK$ and $\ve{d}_i\in\AH$ from~\eqref{Equ:UniqueExpand} and define
$$L(\Sv_{\ve{a}_1}\dots \Sv_{\ve{a}_1})=c_1\,\Sv_{\ve{d}_1}+\cdots+c_m\,\Sv_{\ve{d}_m}.$$
By~\eqref{Equ:DegreePreserving} it follows that
$$w(\Sv_{\ve{a}_1}\dots \Sv_{\ve{a}_1})=\max_{i=1,\ldots,m}w(\Sv_{\ve{d}_i}).$$
Finally, we extend $L$ to $\AA^{\AH}$ by linearity. 


\noindent Since~\eqref{eq:quasipro} reflects precisely~(\ref{hsumproduct}), we obtain the following lemma.

\begin{lemma}\label{Lemma:LEquality}
Let $f\in\AA^{\AH}$ and take $\delta\in\NN$ such that $\expr(f)(\lambda)$ has no poles for all $\lambda\in\NN_0$ with $\lambda\geq\delta$. Then for all $\lambda\geq\delta$,
$\expr(f)(\lambda)=\expr(L(f))(\lambda).$
\end{lemma}

\begin{example}\label{Equ:LExpr}
We have that 
\begin{equation}\label{Equ:LOpExp}
L(\Sv_1{}^4)=4 \Sv_{1,3}+6 \Sv_{2,2}+4 \Sv_{3,1}-12 \Sv_{1,1,2}-12 \Sv_{1,2,1}-12 \Sv_{2,1,1}+24 \Sv_{1,1,1,1}-\Sv_4.
\end{equation}
In particular, as already indicated in~\eqref{Exp:AlgRelation} we get 
\begin{align*}
S_1^4(n)=\expr(L(S_1{}^4))(n)=&4 S_{1,3}(n)+6 S_{2,2}(n)+4 S_{3,1}(n)-12 S_{1,1,2}(n)\\
&-12 S_{1,2,1}(n)-12 S_{2,1,1}(n)+24 S_{1,1,1,1}(n)-S_4(n).
\end{align*}
\end{example}

Clearly, we can consider $\VV$ as a subset of $\AA^{\AH}$, i.e., we can equip $\dfield{\AA^{\AH}}{\sigma}$ with the linearization function $L$. Observe that for any $f\in\VV$ we have that $\sigma(f)\in\VV$. In addition, we obtain the following lemma.

\begin{lemma}\label{Lemma:LSigmaCommute}
For any $f\in\AA^{\AH}$, $\sigma(L(f))=L(\sigma(f))$ and $\sigma^{-1}(L(f))=L(\sigma^{-1}(f))$.
\end{lemma}
\begin{proof}
We only give a proof for $\sigma^{-1}(L(f))=L(\sigma^{-1}(f))$ since $\sigma(L(f))=L(\sigma(f))$ follows analogously.
It suffices to prove $\sigma^{-1}(L(f))=L(\sigma^{-1}(f))$ for a monomial $f\in\AA^{\AH}$ since then we can extend the result by linearity. First consider the product of two nested sums (compare ~\eqref{Equ:CombineLambda}): 
let
$$\bar{\lambda}(\hat{a}_1)\bar{\lambda}(\hat{a}_2)=\sum_{j=1}^mr_j\bar{\lambda}(b_j)$$
with $r_j\in\KK$ and $b_j\in A.$ Then
\begin{eqnarray}\label{TwoFactors}
\sigma^{-1}(\Sv_{\hat{a}_1\ve{a}_1}\Sv_{\hat{a}_2\ve{a}_2})
	&=&\sigma^{-1}\left(\Sv_{\hat{a}_1(\ve{a}_1*\hat{a}_2\ve{a}_2)}+\Sv_{\hat{a}_2(\hat{a}_1\ve{a}_1*\ve{a}_2)}-\sum_{j=1}^mr_j\Sv_{b_j(\ve{a}_1*\ve{a}_2)}\right)\nonumber\\
	&=&\sigma^{-1}\left(\Sv_{\hat{a}_1(\ve{a}_1*\hat{a}_2\ve{a}_2)}\right)+\sigma^{-1}\left(\Sv_{\hat{a}_2(\hat{a}_1\ve{a}_1*\ve{a}_2)}\right)-\sum_{j=1}^mr_j\sigma^{-1}\left(\Sv_{b_j(\ve{a}_1*\ve{a}_2)}\right)\nonumber\\
	&=&\Sv_{\hat{a}_1(\ve{a}_1*\hat{a}_2\ve{a}_2)}-\bar{\lambda}(\hat{a}_1)\Sv_{\ve{a}_1*\hat{a}_2\ve{a}_2}+\Sv_{\hat{a}_2(\hat{a}_1\ve{a}_1*\ve{a}_2)}-\bar{\lambda}(\hat{a}_2)\Sv_{\hat{a}_1\ve{a}_1*\ve{a}_2}\nonumber\\
	  &&-\sum_{j=1}^mr_j\left(\Sv_{b_j(\ve{a}_1*\ve{a}_2)}-\bar{\lambda}(b_j)\Sv_{\ve{a}_1*\ve{a}_2}\right)\nonumber\\
	&=&\Sv_{\hat{a}_1\ve{a}_1}\Sv_{\hat{a}_2\ve{a}_2}-\bar{\lambda}(\hat{a}_1)\Sv_{\ve{a}_1*\hat{a}_2\ve{a}_2}-\bar{\lambda}(\hat{a}_2)\Sv_{\hat{a}_1\ve{a}_1*\ve{a}_2}+\sum_{j=1}^mr_j\bar{\lambda}(b_j)\Sv_{\ve{a}_1*\ve{a}_2}\nonumber\\
	&=&\left(\Sv_{\hat{a}_1\ve{a}_1}-\bar{\lambda}(\hat{a}_1)\Sv_{\ve{a}_1}\right)\left(\Sv_{\hat{a}_2\ve{a}_2}-\bar{\lambda}(\hat{a}_2)\Sv_{\ve{a}_2}\right)\nonumber\\
	&=&\sigma^{-1}(\Sv_{\hat{a}_1\ve{a}_1})\sigma^{-1}(\Sv_{\hat{a}_2\ve{a}_2}).
\end{eqnarray}

Now proceed by induction on the number of factors. Assume that the statement holds for $f$ being the product of $k$ factors $\Sv_{\hat{a}_1\ve{a}_1}\Sv_{\hat{a}_2\ve{a}_2} \cdots \Sv_{\hat{a}_k\ve{a}_k}$ and let
$$
\Sv_{\hat{a}_1\ve{a}_1}\Sv_{\hat{a}_2\ve{a}_2} \cdots \Sv_{\hat{a}_k\ve{a}_k}=\sum_{i}c_i \Sv_{\ve{d}_i}.
$$
Now we get
\begin{eqnarray*}
\sigma^{-1}(\Sv_{\hat{a}_1\ve{a}_1}\Sv_{\hat{a}_2\ve{a}_2}\cdots \Sv_{\hat{a}_k\ve{a}_k}\Sv_{\hat{a}_{k+1}\ve{a}_{k+1}})
	&=&\sigma^{-1}\left(\sum_{i}c_i \Sv_{\ve{d}_i}\Sv_{\hat{a}_{k+1}\ve{a}_{k+1}}\right)\\
	&=&\sum_{i}c_i\sigma^{-1}\left(\Sv_{\ve{d}_i}\Sv_{\hat{a}_{k+1}\ve{a}_{k+1}}\right).
\end{eqnarray*}
Using (\ref{TwoFactors}) and the induction hypothesis we conclude
\begin{eqnarray*}
\sigma^{-1}(\Sv_{\hat{a}_1\ve{a}_1}\cdots \Sv_{\hat{a}_k\ve{a}_k}\Sv_{\hat{a}_{k+1}\ve{a}_{k+1}})
	&=&\sum_{i}c_i\sigma^{-1}\left(\Sv_{\ve{d}_i}\right)\sigma^{-1}\left(\Sv_{\hat{a}_{k+1}\ve{a}_{k+1}}\right)\\
	&=&\sigma^{-1}\left(\sum_{i}c_i \Sv_{\ve{d}_i}\right)\sigma^{-1}\left(\Sv_{\hat{a}_{k+1}\ve{a}_{k+1}}\right)\\
	&=&\sigma^{-1}(\Sv_{\hat{a}_1\ve{a}_1})\cdots\sigma^{-1}(\Sv_{\hat{a}_k\ve{a}_k})\sigma^{-1}\left(\Sv_{\hat{a}_{k+1}\ve{a}_{k+1}}\right).
\end{eqnarray*}
\end{proof}

We conclude with the following lemma which will be essential to prove our main result stated in Theorem~\ref{Thm:MainResult}.

\begin{lemma}\label{Lemma:LinearConstant} 
Let $\bar{\VV}:=\KK\,s^{(1)}_1\oplus\dots\oplus\KK\,s^{(1)}_{n_1}\oplus\KK\,s^{(2)}_1\oplus\dots\oplus\KK\,s^{(2)}_{n_2}\oplus\dots$. Then:
\begin{enumerate}
 \item If $f\in\KK[s^{(1)}_1,\dots,s^{(1)}_{n_1}][s^{(2)}_1,\dots,s^{(2)}_{n_2}]\dots$ with $f\notin \KK$, then $L(f)\in\bar{\VV}$.
 \item If $f\in\bar{\VV}$ with $f\neq0$, then $\sigma(f)\neq f$.
\end{enumerate}
 \end{lemma}
\begin{proof}
The first part is immediate. 
Since $\sigma(f)=f$ if and only if $f=\sigma^{-1}(f)$ we are going to prove that for every $f\in\bar{\VV}$ with $f\neq0$ it follows that
$\sigma^{-1}(f)\neq f.$ Note that $\bar{\VV}\cap\KK=\{\}.$
Let $f\in\bar{\VV}$ be of the form
$$
f=c_ 1\Sv_{\hat{a}_1\ve{a}_1}+c_2 \Sv_{\hat{a}_2\ve{a}_2}+\cdots +c_k \Sv_{\hat{a}_k\ve{a}_k}
$$
for some $c_i\in\KK\setminus{\{0\}}.$
Let $b$ be the minimal letter in $\{\hat{a}_1,\hat{a}_2,\dots,\hat{a}_k\}$ and let $\{r_1,r_2,\dots,r_j\}\subseteq\{1,\dots k\}$ such that $\hat{a}_{r_i}=b$ for $i\in\{1,\ldots,j\}$ and 
let $\{u_1,u_2,\dots,u_h\}\subseteq\{1,\dots k\}$ such that $\hat{a}_{u_i}\neq b$ for $i\in\{1,\ldots,h\}$. Then we get
\begin{eqnarray*}
\sigma^{-1}(f)&=&f-\left(c_1\bar{\lambda}(\hat{a}_1)\Sv_{\ve{a}_1}+\cdots +c_k\bar{\lambda}(\hat{a}_k) \Sv_{\ve{a}_k}\right)\\
	      &=&f-\bar{\lambda}(b)\underbrace{\sum_{i=1}^jc_{r_i}\Sv_{\ve{a}_{r_i}}}_{M:=}-\underbrace{\sum_{i=1}^h\bar{\lambda}(\hat{a}_{u_i})c_{u_i}\Sv_{\ve{a}_{u_i}}}_{U:=}.
\end{eqnarray*}
Due to the definition of $b$ there cannot be a cancellation between the summands of $M$ and $U$. Moreover, due to the choice of $b$ we have that $M\neq0.$ Hence $\sigma^{-1}(f)\neq f.$
\end{proof}

\section{The reduced difference ring}\label{Sec:ReduceDR}

We define the reduced difference ring where all relations of $\dfield{\AA^{\AH}}{\sigma}$ are factored out by the quasi-shuffle algebra, i.e., we define
$$I:=\{f\in\AA^{\AH}|\,L(f)=0\}.$$
It is immediate  that $I$ is an ideal of $\AA^{\AH}$ and we can define the quotient ring
$$\EE^{\AH}:=\AA^{\AH}/I.$$
Even more, $I$ is a reflexive difference ideal~\cite{Cohn:65}, i.e., the following property holds: for any $f\in I$ and any $z\in\ZZ$ we have $\sigma^z(f)\in I$. Namely, for any $f\in I$, we have that $L(\sigma(f))=\sigma(L(f))=\sigma(0)=0$ by Lemma~\ref{Lemma:LSigmaCommute}. Hence, $\sigma(f)\in I$; similarly, it follows that $\sigma^{-1}(f)\in I$.  
From this follows that we can construct the ring automorphism
$\fct{\sigma}{\AA^{\AH}/I}{\AA^{\AH}/I}$ with
$\sigma(a+I)=\sigma(a)+I$. In particular, identifying the elements $f\in\KK(n)[x]$ with $f+I$ we can consider $\dfield{\AA^{\AH}/I}{\sigma}$ as a difference ring extension of $\dfield{\KK(n)[x]}{\sigma}$.\\

In the following we will elaborate how this difference ring can be constructed explicitly in an iterative fashion. 
Using the set of all power products $\Pi$ by $$\Pi=\{\Sv_{\ve{a}_1}\Sv_{\ve{a}_2}\dots\Sv_{\ve{a}_r}\,\mid\,r\in\NN\text{ and } \ve{a}_i\in\AH^*\text{ for }1\leq i\leq r\}$$
we obtain the following possibility to generate $I$. 
\begin{lemma}\label{Lemma:IGenerator}
\begin{equation}\label{Equ:IGenerator}
I=\Big\{\sum_{i=1}^l r_i\,(L(p_i)-p_i)
\,\mid\,l\in\NN,\text{ and }r_i\in\KK(n)[x], p_i\in\Pi\text{ for } 1\leq i\leq l \Big\}.
\end{equation}
\end{lemma}
\begin{proof}
 Denote the set on the right hand side by $J$. First let $f\in J$ \ie $f=\sum_{i=1}^l r_i\,(L(p_i)-p_i)$ for some $l\in\NN\text{ and }r_i\in\KK(n)[x], p_i\in\Pi.$ We have $$L(f)=L\left(\sum_{i=1}^l r_i\,(L(p_i)-p_i)\right)=\sum_{i=1}^l L(r_i\,(L(p_i)-p_i))=\sum_{i=1}^l r_i\,(L(p_i)-L(p_i))=0.$$
 Hence $f\in I$ and thus $J\subseteq I$. Now let $f\in I,$ \ie $f=\sum_{i=1}^lr_i\,p_i$ for some $l\in\NN,\ r_i \in\KK(n)[x]$ and $p_i\in\Pi$ with $L(f)=0.$ We define $v_i:=L(p_i)-p_i,$ hence $p_i=L(p_i)-v_i.$ Using this definition we get
\begin{eqnarray*}
 &f&=\sum_{i=1}^lr_i\,p_i=\sum_{i=1}^lr_i(L(p_i)-v_i)=\sum_{i=1}^lr_i\,L(p_i)-\sum_{i=1}^lr_i\,v_i=\sum_{i=1}^lL(r_i\,p_i)-\sum_{i=1}^lr_i\,v_i\\
 &&=L\left(\sum_{i=1}^lr_i\,p_i\right)-\sum_{i=1}^lr_i\,v_i=L(f)-\sum_{i=1}^lr_i\,v_i=0-\sum_{i=1}^lr_i(L(p_i)-p_i).
\end{eqnarray*}
Thus $f\in J$ and therefore $I\subseteq J$. This completes the proof.
\end{proof}

If we define $I_d=I\cap\AA^{\AH}_d$, it follows again that $I_d$ is a reflexive difference ideal of $\AA^{\AH}_d$ and that the quotient ring $\dfield{\EE^{\AH}_d}{\sigma}$ defined as
$$\EE^{\AH}_d:=\AA^{\AH}_d/I_d$$
is a difference ring. Note that
$$\EE^{\AH}_{0}=\KK(n)[x].$$
Even more, we get the chain of difference ring extensions
\begin{equation}\label{Equ:DRTowerReduced}
\dfield{\KK(n)[x]}{\sigma}=\dfield{\EE^{\AH}_0}{\sigma}\leq\dfield{\EE^{\AH}_1}{\sigma}\leq\dfield{\EE^{\AH}_2}{\sigma}\leq\dots\leq\dfield{\EE^{\AH}}{\sigma}.
\end{equation}

In the following we will work out further how the construction from $\dfield{\EE^{\AH}_{d-1}}{\sigma}$ to $\dfield{\EE^{\AH}_d}{\sigma}$ can be carried out explicitly. Since $L$ is weight preserving it follows that $L(p)$ with $p\in\Pi$ and $w(p)=d$ depends linearly on sums with weight less or equal to $d$. 
Hence by Lemma~\ref{Lemma:IGenerator} we get
\begin{align*}
I_d&=\Big\{\sum_{i=1}^l r_i\,(L(p_i)-p_i)\mid  l\in\NN\text{ with } r_i\in \KK(n)[x], p_i\in\Pi \text{ where } w(p_i)\leq d\text{ for } 1\leq i\leq l \Big\}\\
&=\Big\{g\,h+\sum_{i=1}^l r_i\,(L(p_i)-p_i)
\mid  g\in\AA^{\AH}_d, h\in I_{d-1}\text{ and }l\in\NN\text{ with }\\
&\hspace{5cm}r_i\in \KK(n)[x], p_i\in\Pi \text{ where } w(p_i)=d\text{ for } 1\leq i\leq l \Big\}.
\end{align*}
Exploiting this property we can perform the following construction. At the specified weight $d$ we can set up a matrix containing all the relations 
\begin{equation}\label{Equ:AllRelations}
L(p)-p
\end{equation}
with $p\in\Pi$ and $w(p)=d$ where the columns (except for the last one) represent the sums of weight $d$ coming from $L(p)$ and the last column represents the polynomial $p$ which is built by sums of weight less than $d.$ 
Now we can transform the matrix to its reduced row-echelon form. The sums corresponding to the corner elements can be reduced while the other sums remain as variables. 
Exactly this crucial observation has been strongly utilized for harmonic sums in~\cite{Bluemlein2004} in order to derive all algebraic relations induced by the quasi-shuffle algebra. For further results and heavy calculations using \texttt{HarmonicSums} we refer to~\cite{AblingerDiploma,Ablinger2011,AblingerThesis:12,SSum}.

Now we will link this observation with our difference ring construction. Let  
\begin{equation}\label{Equ:Representants}
a^{(d)}_1,\dots,a^{(d)}_{m_d}
\end{equation}
be all these variables of~\eqref{Equ:WeightVariables} which cannot be reduced by the quasi-shuffle algebra, resp.\ by  $L$. Set $t^{(d)}_{i}:=a^{(d)}_{i}+I_d$. Then we get the ring
\begin{equation}\label{Equ:PolyRingForD}
\EE^{\AH}_d=\AA^{\AH}_d/I_{d}=\EE^{\AH}_{d-1}[t^{(d)}_1,\dots,t^{(d)}_{m_d}].
\end{equation}
By construction this ring is a polynomial ring, i.e., the elements $t^{(d)}_i$ are algebraically independent among each other.

Finally, by the above considerations we conclude that 
$$\EE^{\AH}=\KK(n)[x][t^{(1)}_1,\dots,t^{(1)}_{m_1}][t^{(2)}_1,\dots,t^{(2)}_{m_2}]\dots.$$

The above construction will be illustrated by the following example.

\begin{example}\label{Exp:LinearAlgebra}
First, we consider the variables of the harmonic sums (i.e., $\AH=A_h$) of weight~4:
$$
\ve{v}=(\Sv_{1,1,1,1},\Sv_{1,1,2},\Sv_{1,2,1},\Sv_{2,1,1},\Sv_{2,2},\Sv_{1,3},\Sv_{3,1},\Sv_4).
$$
Here we consider all non-trivial polynomials~\eqref{Equ:AllRelations} for $p\in\Pi$ with $w(p)=4$:
\begin{eqnarray*}
    R&=&\{ 4 \Sv_{1,3}+6 \Sv_{2,2}+4 \Sv_{3,1}-12 \Sv_{1,1,2}-12 \Sv_{1,2,1}-12 \Sv_{2,1,1}+24 \Sv_{1,1,1,1}-\Sv_4-\Sv_1{}^4,\\
    &&-2 \Sv_{1,3}-2 \Sv_{2,2}-2 \Sv_{3,1}+2 \Sv_{1,1,2}+2 \Sv_{1,2,1}+2 \Sv_{2,1,1}+\Sv_4-\Sv_2 \Sv_1{}^2,\\
    &&\Sv_{1,3}+\Sv_{3,1}-\Sv_4-\Sv_1 \Sv_3,\\
    &&2 \Sv_{2,2}-\Sv_4-\Sv_2{}^2,\\
    &&\Sv_{2,2}-2 \Sv_{1,1,2}-2 \Sv_{1,2,1}-2\Sv_{2,1,1}+6 \Sv_{1,1,1,1}-\Sv_{1,1}{}^2,\\
    &&-\Sv_{1,3}-\Sv_{3,1}+\Sv_{1,1,2}+\Sv_{1,2,1}+\Sv_{2,1,1}-\Sv_2 \Sv_{1,1},\\
    &&-\Sv_{1,3}-\Sv_{2,2}+2 \Sv_{1,1,2}+\Sv_{1,2,1}-\Sv_1 \Sv_{1,2},\\
    &&-\Sv_{2,2}-\Sv_{3,1}+\Sv_{1,2,1}+2 \Sv_{2,1,1}-\Sv_1\Sv_{2,1}\}.
\end{eqnarray*}
E.g., $p=\Sv_1{}^4\in\Pi$ yields first entry of $R$. Note that by Lemma~\ref{Lemma:LEquality} the elements in $R$ transformed to \AH-sum expressions evaluate to $0$ for all $n\geq0$. E.g.,
the first entry is precisely the equation worked out in Example~\ref{Equ:LExpr}.
Note further that
$$I_4=\{f\,h+\sum_{r\in R}k_{r}\,r|f\in\AA^{\AH}_4, h\in I_3\text{ and }k_r\in\KK(n)[x]\}.$$
The relations in $R$ lead to the matrix 
$$\left(M|h\right)=
\left(
\begin{array}{cccccccc|c}
 24 & -12 & -12 & -12 & 6 & 4 & 4 & -1 & -\Sv_1{}^4 \\
 0 & 2 & 2 & 2 & -2 & -2 & -2 & 1 & -\Sv_1{}^2 \Sv_2 \\
 0 & 0 & 0 & 0 & 0 & 1 & 1 & -1 & -\Sv_1 \Sv_3 \\
 0 & 0 & 0 & 0 & 2 & 0 & 0 & -1 & -\Sv_2{}^2 \\
 6 & -2 & -2 & -2 & 1 & 0 & 0 & 0 & -\Sv_{1,1}{}^2 \\
 0 & 1 & 1 & 1 & 0 & -1 & -1 & 0 & -\Sv_2 \Sv_{1,1} \\
 0 & 2 & 1 & 0 & -1 & -1 & 0 & 0 & -\Sv_1 \Sv_{1,2} \\
 0 & 0 & 1 & 2 & -1 & 0 & -1 & 0 & -\Sv_1 \Sv_{2,1} \\
\end{array}
\right)
$$
where $\left(M|h\right)\,\ve{v}^{t}$ gives back a vector with the entries from $R$.
Now we reduce this matrix using Gaussian elimination together with some relations of lower weight to get the reduced echelon form
$$
\left(M'|h'\right)=\left(
\begin{array}{cccccccc|c}
 1 & 0 & 0 & 0 & 0 & 0 & 0 & -\frac{1}{4} & -\frac{1}{24} \Sv_{1}{}^4-\frac{1}{4} \Sv_{2} \Sv_{1}{}^2-\frac{1}{3} \Sv_{3} \Sv_{1}-\frac{1}{8}
   \Sv_{2}{}^2 \\
 0 & 1 & 0 & -1 & 0 & 0 & 1 & -\frac{1}{2} & \Sv_{1} \left(\Sv_{2,1}-\Sv_{3}\right)-\frac{1}{2} \Sv_{1}{}^2 \Sv_{2} \\
 0 & 0 & 1 & 2 & 0 & 0 & -1 & -\frac{1}{2} & -\frac{1}{2} \Sv_{2}{}^2-\Sv_{1} \Sv_{2,1} \\
 0 & 0 & 0 & 0 & 1 & 0 & 0 & -\frac{1}{2} & -\frac{1}{2} \Sv_{2}{}^2 \\
 0 & 0 & 0 & 0 & 0 & 1 & 1 & -1 & -\Sv_{1} \Sv_{3} \\
 0 & 0 & 0 & 0 & 0 & 0 & 0 & 0 & 0 \\
 0 & 0 & 0 & 0 & 0 & 0 & 0 & 0 & 0 \\
 0 & 0 & 0 & 0 & 0 & 0 & 0 & 0 & 0 \\ 
\end{array}
\right).
$$
Set $R':=\left(M'|h'\right)\,\ve{v}^{t}$. By construction the elements of $R$ can be written in terms of the elements of $R'$ plus some extra elements from $I_3$. Thus the following holds:
$$I_4=\{f\,h+\sum_{r\in R'}k_{r}\,r|f\in\AA^{\AH}_4, h\in I_3\text{ and }k_r\in\KK(n)[x]\}.$$
The advantage of this representation of $I_4$ with $R'$ (instead of $R$) is that one can read off rewrite rules to eliminate sums:
the variables corresponding to the corner points of the reduced matrix, i.e., the variables 
\begin{equation}\label{Equ:ReducedHarSumVar}
\Sv_{1,1,1,1},\Sv_{1,1,2},\Sv_{1,2,1},\Sv_{2,2},\Sv_{1,3}
\end{equation}
can be reduced with the substitution rules\footnote{Reinterpreting the variables as \AH-sum expressions, the left hand sides equal the right hand sides for all $n\geq0$.}
\begin{equation}\label{Equ:Substitution4}
\begin{split}
\Sv_{1,1,1,1}&\to \frac{\Sv_4}{4}+\frac{1}{24} \Sv_{1}{}^4+\frac{1}{4} \Sv_{2} \Sv_{1}{}^2+\frac{1}{3} \Sv_{3} \Sv_{1}+\frac{1}{8}\Sv_{2}{}^2,\\
\Sv_{1,1,2}&\to \Sv_{2,1,1}-\Sv_{3,1}+\frac{\Sv_4}{2} - \Sv_{1} \left(\Sv_{2,1}-\Sv_{3}\right)+\frac{1}{2} \Sv_{1}{}^2 \Sv_{2},\\
\Sv_{1,2,1}&\to -2 \Sv_{2,1,1}+\Sv_{3,1}+\frac{\Sv_4}{2}+\frac{1}{2} \Sv_{2}{}^2+\Sv_{1} \Sv_{2,1},\\
\Sv_{2,2}&\to \frac{\Sv_4}{2} +\frac{1}{2} \Sv_{2}{}^2 ,\\
\Sv_{1,3}&\to -\Sv_{3,1}+\Sv_4+\Sv_{1} \Sv_{3}
\end{split}
\end{equation}
while the $\Sv_{2,1,1},\Sv_{3,1},\Sv_4$
remain as variables\footnote{Note that a different order of sums in $\ve{v}$ would lead to a different set of remaining variables; certain orders would lead to a basis built by Lyndon words, compare~\cite{Hoffman1992,Hoffman1997,Hoffman}.}. In other words, within the ring $\EE^{\AH}_4=\AA^{\AH}_4/I_4$ these elements cannot be eliminated further and we arrive at the polynomial ring
$$\EE^{\AH}_4=\AA^{\AH}_4/I_4=\EE^{\AH}_3[\Sv_{2,1,1}+I_4,\Sv_{3,1}+I_4,\Sv_4+I_4]$$
with the transcendental generators $\Sv_{2,1,1}+I_4,\Sv_{3,1}+I_4,\Sv_4+I_4$.\\
We remark that we can calculate with the representants $\Sv_{2,1,1},\Sv_{3,1},\Sv_4$ instead of the equivalence classes $\Sv_{2,1,1}+I_4,\Sv_{3,1}+I_4,\Sv_4+I_4$, respectively. There is only a subtle detail. If we want to perform $\sigma(f)$ for some $f\in\AA^{\AH}/I$, we apply first $\sigma(f)$ where $\sigma$ is taken from $\dfield{\AA^{\AH}}{\sigma}$. Since the sums~\eqref{Equ:ReducedHarSumVar} might be again introduced by $\sigma$, we have to eliminate them by the found substitution rules stated in~\eqref{Equ:Substitution4}.

Similarly, one can construct the polynomial ring
$$\EE^{\AH}_3=\EE^{\AH}_2[\Sv_{2,1}+I_3,\Sv_{3}+I_3]$$
with the substitution rules
\begin{equation}\label{Equ:Substitution3}
\Sv_{1,2}\to -\Sv_{2,1}
+\Sv_1 \Sv_2+\Sv_3\quad\text{and}\quad
\Sv_{1,1,1}\to \frac{1}{6} \Sv_1^3
+\frac{1}{2} \Sv_1 \Sv_2
+\frac{1}{3} \Sv_3,
\end{equation}
the polynomial ring
$$\EE^{\AH}_2=\EE^{\AH}_1[\Sv_{2}+I_2]$$
with the substitution rule
\begin{equation}\label{Equ:Substitution2}
\Sv_{1,1}\to \frac{1}{2} \Sv_1^2
+\frac{1}{2} \Sv_2,
\end{equation}
and the polynomial ring
$$\EE^{\AH}_1=\KK(n)[x][\Sv_{1}+\{0\}].$$
To sum up, we obtain the polynomial ring
$$\EE^{\AH}_4=\KK(n)[x][\Sv_{1}+\{0\}][\Sv_{2}+I_2][\Sv_{2,1}+I_3,\Sv_{3}+I_3][\Sv_{2,1,1}+I_4,\Sv_{3,1}+I_4,\Sv_4+I_4].$$
Now define the function
$\fct{\rho}{\AA^{\AH}_4}{\bar{\EE}^{\AH}_4}$ where $\rho(f')$ is obtained by applying the substitution rules of~\eqref{Equ:Substitution4}, \eqref{Equ:Substitution3} and~\eqref{Equ:Substitution2}
such that one obtains an element of the polynomial ring
$$\bar{\EE}^{\AH}_4=\KK(n)[x][\Sv_{1}][\Sv_{2}][\Sv_{2,1},\Sv_{3}][\Sv_{2,1,1},\Sv_{3,1},\Sv_4]$$
being a subring of $\AA^{\EE}$. Then
the difference ring $\dfield{\EE^{\AH}_4}{\sigma}$ can be represented by the difference ring $\dfield{\bar{\EE}^{\AH}_4}{\bar{\sigma}}$ with the ring automorphism $\fct{\bar{\sigma}}{\bar{\EE}^{\AH}_4}{\bar{\EE}^{\AH}_4}$ defined by $\bar{\sigma}(f)=\rho(\sigma(f))$. Here we apply first the given automorphism $\fct{\sigma}{\AA^{\AH}}{\AA^{\AH}}$ with $f':=\sigma(f)\in\AA^{\AH}$ and apply afterwards the constructed substitution rules to get the element $\rho(f')\in\bar{\EE}^{\AH}_4$.
\end{example}

In general, the representants in~\eqref{Equ:Representants} for the variables of weight $d$ (plus the other representants with weights $<d$ from the previous construction steps) can be used to express all the elements of $f\in\EE^{\AH}_d$. In particular, all the ring operations can be simply performed in this polynomial ring. Only if one applies $\sigma(f)$, one has to apply the corresponding substitution rules in order to rewrite the expression again in terms of the representants. 

\begin{remark}\label{Remark:IncreaseAlphabet}
So far, we fixed $\AH\in\{A_h,A_a,A_{c(M)}\}$ in~\eqref{Equ:ChooseH} and performed the construction for the chosen \AH-sums to get the difference rings $\dfield{\EE^{\AH}_d}{\sigma}$ and $\dfield{\EE^{\AH}}{\sigma}$. Now we fix $\AH=A_{c(M)}$ and look what happens if we vary the finite set $M\in\NN$.
More precisely, if 
$$
M_1\subset M_2\subset M_3 \subset \dots
$$
is a chain of finite subsets of $\NN$ then 
$$
A_{c(M_1)} \subset A_{c(M_2)} \subset \dots \subset A_{c(M_k)} \subset \dots \subset A_c.
$$
Furthermore, if $B^{(d)}_i$ denotes a set of representants of weight $d$ for $A_{c(M_i)}$ as defined in \eqref{Equ:Representants}, then by using our construction above we can find representants such that
$$
B^{(d)}_1 \subset B^{(d)}_2 \subset B^{(d)}_3 \dots
$$
As a consequence, it follows that $\EE_d^{A_{c(M_r)}}$ is a polynomial ring extension of $\EE_d^{A_{c(M_{r-1})}}$ where the transcendental generators are the elements $\cup_{1\leq j\leq d} V_j$ with $V_j:=B^{(j)}_r\setminus B^{(j)}_{r-1}$; here $V_j$  are the generators that represent the additional \AH-sums of weight $j$. Moreover, $\dfield{\EE_d^{A_{c(M_r)}}}{\sigma}$ is a difference ring extension of $\dfield{\EE_d^{A_{c(M_{r-1})}}}{\sigma}$.
\end{remark}

We remark that any element from the equivalence class $a^{(d)}_i+I_d$ delivers the same evaluation expression.

\begin{lemma}\label{Lemma:ClassIstheSame}
Let $f,g\in t^{(d)}_i=a_i^{(d)}+I_d$ and take $\delta\in\NN_0$ such that for all $\delta\leq\lambda\in\NN_0$, $\expr(f)(\delta)$ and $\expr(g)(\delta)$ do not have poles. Then for all $\lambda\geq\delta$,
$\expr(f)(\lambda)=\expr(g)(\lambda)$.
\end{lemma}
\begin{proof}
Let $h_1,h_2\in I_d$ such that $f=a^{(d)}_i+h_1$ and $g=a^{(d)}_i+h_2$.
Since $\expr(a^{(d)}_i)$ and $\expr(h_i)$ have no poles, $\expr(f)(\lambda)$ and $\expr(g)(\lambda)$ have no poles for $\lambda\geq\delta$. Therefore 
$$\expr(f)(\lambda)-\expr(g)(\lambda)=\expr(f-g)(\lambda)=\expr(h_1-h_2)(\lambda)=\expr(h_1)(\lambda)-\expr(h_2)(\lambda)$$ 
for all $\lambda\geq\delta$. 
By Lemma~\ref{Lemma:LEquality} we conclude that $\expr(h_i)(\lambda)=\expr(L(h_i))(\lambda)=0$ and so $\expr(f)(\lambda)=\expr(g)(\lambda)$.
\end{proof}

In addition, we obtain the following connection with the linearization operator.
\begin{lemma}
For any $f\in a^{(d)}_i+I$ we have that $L(f)=L(a^{(d)}_i)$.
\end{lemma}
\begin{proof}
Write $f=a^{(d)}_i+h$ with $h\in I$. Then $L(f)=L(a^{(d)}_i+h)=L(a^{(d)}_i)+L(h)=L(a^{(d)}_i)$.
\end{proof}

This implies that we can define the linearization map $\fct{\bar{L}}{\EE^{\AH}}{\VV}$ with $$\bar{L}(g+I)=L(g)$$ 
for $g\in\AA^{\AH}$. In addition, we get the following result.

\begin{lemma}\label{Lemma:LBarSigmaCommute}
For any $f\in\EE^{\AH}$, $\sigma(\bar{L}(f))=\bar{L}(\sigma(f))$.
\end{lemma}
\begin{proof}
Write $f=g+h$ with $g\in\AA^{\AH}$ and $h\in I$. By definition, $L(h)=0$. Since $I$ is a difference ideal, $\sigma(h)\in I$. 
Together with Lemma~\ref{Lemma:LSigmaCommute} we get
$$\sigma(\bar{L}(f))=\sigma(\bar{L}(g+h))=\sigma(L(g))=L(\sigma(g))=\bar{L}(\sigma(g)+\sigma(h))=\bar{L}(\sigma(g+h))=\bar{L}(\sigma(f)).$$ 
\end{proof}

\section{Structural results in \rpisiSE-ring}\label{Sec:pisiRings}

In the following we will bring in addition difference ring theory. 
Here we will refine the naive construction given by Lemma~\ref{Lemma:SumExtension} that we used to construct the difference ring $\dfield{\AA}{\sigma}$ in Section~\ref{Sec:BasicDR}. Namely, we will improve this type of extensions by the so-called \sigmaSE-extensions.

\begin{definition}
Take the difference ring extension $\dfield{\AA[t]}{\sigma}$ of $\dfield{\AA}{\sigma}$ from Lemma~\ref{Lemma:SumExtension}. Then this extension is called \sigmaSE-extension if $\const{\AA[t]}{\sigma}=\const{\AA}{\sigma}$.\\
In particular, $\dfield{\AA[t_1]\dots[t_e]}{\sigma}$ is called a (nested) \sigmaSE-extension of $\dfield{\AA}{\sigma}$ if $\dfield{\AA[t_1]\dots[t_i]}{\sigma}$ is a \sigmaSE-extension of $\dfield{\AA[t_1]\dots[t_{i-1}]}{\sigma}$ for all $1\leq i\leq e$. 
\end{definition}

In order to verify that $\const{\AA[t]}{\sigma}=\const{\AA}{\sigma}$ holds, there is the following equivalent characterization; see~\cite[Thm.~2.12]{DR1} or~\cite[Thm.~3]{DR2}. Note that this result is a generalization/refinement of Karr's difference field theory~\cite{Karr:81,Karr:85}.

\begin{theorem}\label{Thm:NotSigmaExt}
Take the difference ring extension $\dfield{\AA[t]}{\sigma}$ of $\dfield{\AA}{\sigma}$ from Lemma~\ref{Lemma:SumExtension} with $\sigma(t)=t+\beta$ ($\beta\in\AA$) where $\const{\AA}{\sigma}$ is a field.  Then this is a \sigmaSE-extension (i.e., $\const{\AA[t]}{\sigma}=\const{\AA}{\sigma}$) iff there is no $g\in\AA$ with 
\begin{equation}\label{Equ:Tele}
\sigma(g)=g+\beta.
\end{equation}
\end{theorem}

Exactly this result is the driving property within the summation package~\texttt{Sigma}~\cite{Sigma1,Sigma2}. Indefinite nested sums are represented not in the naive construction given in Lemma~\ref{Lemma:SumExtension}, but sums are only adjoined if the constants remain untouched, i.e., if the telescoping problem cannot be solved. More precisely, in \texttt{Sigma} indefinite summation algorithms are implemented~\cite{FastAlgorithm1,FastAlgorithm2,FastAlgorithm3} that decide efficiently, if there exists a $g\in\AA$ such that~\eqref{Equ:Tele} holds. If such an element does not exist, we can perform the construction as stated in Lemma~\ref{Lemma:SumExtension} and we have constructed a \sigmaSE-extension without extending the constant field. Otherwise, if we find such a $g$, we can use $g$ itself to represent the ``sum'' with the shift-behaviour~\eqref{Equ:Tele}. In this way, the occurring sums within an expression are rephrased step by step by a tower of \sigmaSE extensions\footnote{We remark that one can hunt for recurrences of definite sums and can solve recurrence in terms d'Alembertian solution in such difference rings; for details see, e.g.,~\cite{Sigma2} and references therein.} of the form~\eqref{Equ:NaiveChainDR} with
$$\KK=\const{\AA_0}{\sigma}=\const{\AA_1}{\sigma}=\const{\AA_2}{\sigma}=\dots=\const{\AA}{\sigma}.$$
However, this construction gets more and more expensive, the more variables are adjoined. In contrast to that, if we restrict to \AH-sums, the construction of the tower of ring extensions~\eqref{Equ:DRTowerReduced}, which relies purely on a linear algebra engine, is by far more efficient; see Example~\ref{Exp:LinearAlgebra}.

\begin{example}\label{Exp:HarmonicSum7}
Given such a reduction in terms of basis sums (using linear algebra), we could verify with the summation package \texttt{Sigma} that the difference ring extension $\dfield{\EE^{\AH}_d}{\sigma}$ of $\dfield{\KK(n)[x]}{\sigma}$ with $\AH=H_a$ and $d=1,2,3,4,5,6,7$ is a \sigmaSE-extension. For $d=7$ we represented 507 basis sums, i.e., we constructed a tower of 507 \sigmaSE-extensions within 5 days. The case $d=8$ is currently out of scope.
Similarly, we could show for $\AH=A_{c(M)}$ with various finite subsets $M\subset\NN$ and $d\in\NN$ that $\dfield{\EE^{\AH}_d}{\sigma}$ forms a \sigmaSE-extension of $\dfield{\KK(n)[x]}{\sigma}$. 
\end{example}

\noindent Note that Theorem~\ref{Thm:NotSigmaExt} can be supplemented by the following structural result; for the difference field version see, e.g.,~\cite[Thm.~2.4]{AlgebraicDF}.

\begin{proposition}\label{Prop:TeleStructure}
Let $\dfield{\AA[t_1]\dots[t_e]}{\sigma}$ be a \sigmaSE-extension of $\dfield{\AA}{\sigma}$ with $\sigma(t_i)-t_i\in\AA$ where $\KK:=\const{\AA}{\sigma}$ is a field. Let $\beta\in\AA$ and let $g\in\AA[t_1]\dots[t_e]$ be a solution of~\eqref{Equ:Tele}. Then $g=\sum_{i=1}^e c_i\,t_i+w$ with $c_i\in\KK$ and $w\in\AA$.
\end{proposition}
\begin{proof}
We will show the result by induction on $e$. Suppose that the proposition holds for $e-1$ such extensions. Now consider the case of $e$ extensions as stated in the proposition. Set $\GG:=\AA[t_1]\dots[t_{e-1}]$ and define $\beta_e:=\sigma(t_e)-t_e\in\AA$. Let $\beta\in\AA$ and $g\in\GG[t_e]$ such that~\eqref{Equ:Tele} holds. 
By~\cite[Lemma~6]{DR2} (resp.~\cite[Lemma~7.2]{DR1}) it follows that 
$g=c\,t_e+\tilde{g}$ for some $c,\tilde{g}\in\GG$. Comparing coefficients w.r.t.\ $t_e$ in $\sigma(c\,t_e+\tilde{g})-(c\,t_e+\tilde{g})=\beta$ implies that $\sigma(c)-c=0$, i.e., $c\in\const{\GG[t_e]}{\sigma}=\KK$. Therefore we get
$\sigma(\tilde{g})-\tilde{g}=\beta-(\sigma(c\,t_e)-c\,t_e)=\beta-c\,\beta_e\in\AA$. By the induction assumption we conclude that $\tilde{g}=\sum_{i=1}^{e-1}c_i\,t_i+w$ with $c_i\in\KK$ and $w\in\AA$. Hence $g=c\,t_e+\sum_{i=1}^{e-1}c_i\,t_i+w$ and the proposition is proven.
\end{proof}

\noindent In the following we will work out further properties in such a tower of \sigmaSE-extensions. This will finally enable us to show us in one stroke that the tower~\eqref{Equ:DRTowerReduced} is built by an infinite tower of \sigmaSE-extensions, i.e., that
$$\const{\EE^{\AH}}{\sigma}=\const{\AA^{\AH}/I}{\sigma}=\KK$$
for $\AH\in\{A_h,A_a,A_{c(M)}\}$ with a finite set $M\in\NN$.
In other words, for the special class of \AH-sums, we can dispense the user from all the difference ring algorithms of \texttt{Sigma}.

In the following let $\dfield{\KK(n)[x]}{\sigma}$ be our difference ring with $\sigma(n)=n+1$, $\sigma(x)=-x$ and $x^2=1$ where
$$\const{\KK(n)[x]}{\sigma}=\KK.$$
Let $f,g\in\KK[n]\setminus\{0\}$. Note that $\sigma^k(g)\in\KK[n]$. Hence we can define the dispersion~\cite{Abramov:71}
$$\disp(f,g)=\max\{k\geq0|\gcd(f,\sigma^k(g))\neq1\}.$$
Let $q_1,\dots,q_r\in\KK[n]\setminus\{0\}$ be polynomials. Then we define the set 
$$\KK[n]_{\{q_1,\dots,q_r\}}:=\{\frac{a}{q_1^{m_1}\dots q_r^{m_r}}|\, (m_1,\dots,m_r)\in\NN_0^r\setminus\{\vect{0}\};a\in\KK[n]; \gcd(q_1\cdots q_r,a)=1\}.$$

With these notions, we restrict the class of \sigmaSE-extensions as follows.

\begin{definition}
Let $\dfield{\KK(n)[x][s_1]\dots[s_e]}{\sigma}$ be a \sigmaSE-extension  of $\dfield{\KK(n)[x]}{\sigma}$ and let $P=\{p_1,\dots,p_r\}$ with $p_i\in\KK[n]\setminus\{0\}$. 
Then the extension is called $P$-normalized, if for any $1\leq i\leq e$ we have that
$\sigma(s_i)-s_i\in\KK[n]_P[x][s_1\dots,s_{i-1}].$
\end{definition}

\begin{example} 
As worked out in Example~\ref{Exp:HarmonicSum7} (compare also  Theorem~\ref{Thm:MainResult} below),
the difference ring $\dfield{\EE^{\AH}_4}{\sigma}$ (resp.\  $\dfield{\bar{\EE}^{\AH}_4}{\bar{\sigma}}$) of Example~\ref{Exp:LinearAlgebra}
is a \sigmaSE-extension of $\dfield{\KK(n)[x]}{\sigma}$. In particular, it is a $\{n+1\}$-normalized \sigmaSE-extension.
\end{example}

In the remaining subsection we will show in Proposition~\ref{Prop:QRelation}
that the telescoping solution $g$ of~\eqref{Equ:Tele} will not depend on $n$ provided that one restricts to certain normalized \sigmaSE-extensions 
and takes $f$ with a particular shape.

\begin{lemma}\label{Lemma:NoSolutionInK(n)}
Consider the difference field $\dfield{\KK(n)}{\sigma}$ with $n$ being transcendental over $\KK$ and with $\sigma(n)=n+1$. Let $q\in\KK[n]\setminus\KK$ with $\disp(q,q)=0$, let $u\in\KK^*$ and let $v\in\KK[n]$ with $\gcd(v,q)=1$. Then there is no $g\in\KK(n)$ such that $\sigma(g)+u\,g=\frac{v}{q}$.
\end{lemma}
\begin{proof}
Suppose that there is such a $g\in\KK(n)$.
By Abramov's universal denominator bounding~\cite{Abramov:89a} or Bronstein's generalization\footnote{For a unified formulation see Theorem~2 in~\cite{Schneider:04b}.}~\cite{Bron:00} (see Theorem~8 and~10 therein) it follows that $g\in\KK[n]$. Hence $\sigma(g)+u\,g\in\KK[n]$, a contradiction that $q\notin\KK$.
\end{proof}

\begin{lemma}\label{Lemma:NoSolutioninKnx}
Consider the difference ring $\dfield{\KK(n)[x]}{\sigma}$ with $n$ being transcendental over $n$ and $\sigma(n)=n+1$ and with $\sigma(x)=-x$ and $x^2=1$.
Let $p\in\KK[n]\setminus\KK$ with $\disp(p,p)=0$ and let $f\in\KK[n]_{\{p\}}[x]$ with $f\neq0$. Then there is no $g\in\KK(n)[x]$ with $\sigma(g)-g=f$.
\end{lemma}
\begin{proof}
Let $g=g_0+g_1\,x^1$ with $g_i\in\KK(n)$ and $f=f_0+f_1\,x^1$ with $f_i\in\KK[n]_{\{p\}}$ such that $\sigma(g)-g=f$. Since $f\notin\KK[x]$, there is a $j\in\{0,1\}$ such that $f_j\not\in\KK$. By coefficient comparison we get $\sigma(g_j)-(-1)^j\,g_j=(-1)^jf_j$. Write $f_j=\frac{v}{q}$ with $\gcd(v,q)=1$. Note that $q=c\,p^m$ for some $m\in\NN$ and $c\in\KK^*$. Since $\disp(p^m)=0$, the existence of the solution $g_j$ contradicts to Lemma~\ref{Lemma:NoSolutionInK(n)}. 
\end{proof}

\begin{lemma}\label{Lemma:DeltaIsClosed}
Let $\dfield{\KK(n)[x]}{\sigma}$ be the difference ring from Lemma~\ref{Lemma:NoSolutioninKnx}. Let $P=\{p_1,\dots,p_r\}$ with $p_i\in\KK[n]\setminus\{0\}$ and let $\dfield{\EE}{\sigma}$ be a $P$-normalized \sigmaSE-extension of $\dfield{\KK(n)[x]}{\sigma}$. Set $\HH:=\KK(n)_P[x][s_1]\dots[s_e]$. Then the following holds.
\begin{enumerate}
\item For any $f,g\in\HH$, $f+g\in\HH$.
\item For any $f\in\KK[s_1,\dots,s_e]$, $\sigma(f)-f\in\HH$.
\end{enumerate}
\end{lemma}
\begin{proof}
The first statement is immediate. Now let $f=\sum_{(n_1,\dots,n_e)\in\NN_0^e} c_{(n_1,\dots,n_e)}s_1^{n_1}\dots s_e^{n_e}$ with $c_{(n_1,\dots,n_e)}\in\KK$ and set
$\beta_i:=\sigma(s_i)-s_i\in\HH$. Then by expanding $\sigma(f)$ into its monomials  we get
\begin{align*}
\sigma(f)=&\sum_{(n_1,\dots,n_e)\in\NN_0^e} c_{(n_1,\dots,n_e)}(s_1+\beta_1)^{n_1}\dots (s_e+\beta_e)^{n_e}\\
=&f+\sum_{(n_1,\dots,n_e)\in\NN_0^e\setminus\{\vect{0}\}} d_{(n_1,\dots,n_e)}s_1^{n_1}\dots s_e^{n_e}
\end{align*}
where for each $(n_1,\dots,n_e)\in\NN\setminus\{\vect{0}\}$ we have that $d_{(n_1,\dots,n_e)}\in\KK(n)_P[x]$. This last statement follows from the fact that at least one $\beta_i\in\HH$ is a factor of $d_{(n_1,\dots,n_e)}$.
Since a sum of elements of $\HH$ is again from $\HH$, we conclude that $\sigma(f)-f\in\HH$.  
\end{proof}

\begin{proposition}\label{Prop:QRelation}
Let $\dfield{\KK(n)[x]}{\sigma}$ as given in Lemma~\ref{Lemma:NoSolutioninKnx}. Let $P=\{p_1,\dots,p_r\}$ with $p_i\in\KK[n]\setminus\{0\}$ and $\disp(p_i,p_j)=0$ for all $1\leq i,j\leq r$.
Let $\dfield{\KK(n)[x][s_1]\dots[s_e]}{\sigma}$ be a $P$-normalized \sigmaSE-extension of $\dfield{\KK(n)[x]}{\sigma}$. If $f\in\KK[n]_P[x][s_1\dots,s_e]$  and $g\in\KK(n)[x][s_1]\dots[s_e]$ such that $\sigma(g)-g=f$ holds, then $g\in\KK[s_1,\dots,s_e]$.
\end{proposition}
\begin{proof}
We prove the statement
by induction on the number of extensions.  If $e=0$, the situation $f\neq0$ cannot occur by Lemma~\ref{Lemma:NoSolutioninKnx}. If $f=0$, $g\in\KK$ and the base case is proven.\\ 
Now suppose that the theorem holds for $e$ extensions and
consider the \sigmaSE-extension
$\dfield{\HH[s]}{\sigma}$ of $\dfield{\KK(n)[x]}{\sigma}$ with $\HH=\KK(n)[x][s_1]\dots[s_e]$ and $\sigma(s)=s+\beta$ with
$\beta\in\HH$. Denote $\HH_P=\KK[n]_P[x][s_1]\dots[s_e]$. Now suppose that
$\sigma(g)-g=f$ holds where
$g\in\HH[s]$ and
$f\in\HH_P[s]$. 
Again, if $f=0$, $g\in\KK$ and we are done.\\ 
Finally, suppose that $f\neq0$ holds.
Write
$g=\sum_{i=0}^d g_is^i$ with $g_i\in\HH$. If $g_i\in\KK[s_1,\dots,s_e]$ for all $i$, we are done. Otherwise, let $j$ be maximal such that $g_j\notin\KK[s_1,\dots,s_e]$.
Define
$g':=\sum_{i=0}^{j}g_i\,s^i\in\HH[s]$ and $\gamma:=\sum_{i=j+1}^{d}g_i\,s^i\in\KK[s_1]\dots[s_e][s]$, i.e., $g=g'+\gamma$. By Lemma~\ref{Lemma:DeltaIsClosed} it follows that $h:=\sigma(\gamma)-\gamma\in\HH_P[s]$.
In particular, since $f\in\HH_P[s]$, $f':=f-h\in\HH_P[s]$ by Lemma~\ref{Lemma:DeltaIsClosed}. By construction,
$$\sigma(g')-g'=f'$$
with $f'=\sum_{i=0}^j f'_j\,s^i$ for some $f'_j\in\HH_P$.
In particular, by leading coefficient comparison, 
$$\sigma(g_j)-g_j=f'_j.$$
By the induction assumption it follows that $g_j\in\KK[s_1,\dots,s_e]$, a contradiction. 
\end{proof}

\section{Algebraic independence of (cyclotomic) harmonic sums}\label{Sec:MainResult}

First, we define the ring of sequences and the notion of difference ring embeddings.
Consider the set of sequences
$\KK^{\NN_0}$ with elements $\langle a_{\lambda}\rangle_{\lambda\geq0}=\langle
a_0,a_1,a_2,\dots\rangle$, $a_i\in\KK$. With component-wise addition and multiplication we obtain a commutative ring; the field $\KK$ can be
naturally embedded by identifying $k\in\KK$ with the sequence
$\langle k,k,k,\dots\rangle$; we write $\vect{0}=\langle0,0,0,\dots\rangle$.

\noindent We follow the construction
from~\cite[Sec.~8.2]{AequalB} in order to turn the
shift
\begin{equation}\label{Equ:ShiftOp}
\shiftS:{\langle a_0,a_1,a_2,\dots\rangle}\mapsto{\langle
a_1,a_2,a_3,\dots\rangle}
\end{equation}
into an automorphism: We define an equivalence relation
$\sim$ on $\KK^{\NN_0}$ by $\langle a_{\lambda}\rangle_{\lambda\geq0}\sim
\langle b_\lambda\rangle_{\lambda\geq0}$ if there exists a $d\geq 0$ such that
$a_k=b_k$ for all $k\geq d$. The equivalence classes form a ring
which is denoted by $\seqK$; the elements of $\seqK$ (also called germs) will be
denoted, as above, by sequence notation.
Now it is immediate that
$\fct{\shiftS}{\seqK}{\seqK}$ with~\eqref{Equ:ShiftOp} forms a ring automorphism. The difference ring $\dfield{\seqK}{\shiftS}$ is called the ring of sequences (over $\KK$).

A difference ring homomorphism $\fct{\tau}{\GG_1}{\GG_2}$ between difference
rings $\dfield{\GG_1}{\sigma_1}$ and $\dfield{\GG_2}{\sigma_2}$ is a
ring homomorphism such that $\tau(\sigma_1(f))=\sigma_2(\tau(f))$ for all
$f\in\GG_1$. If $\tau$ is injective, we call $\tau$ a difference ring monomorphism or a difference ring embedding. In this regard note that $\dfield{\tau(\GG_1)}{\sigma}$ forms a difference ring which is the same as $\dfield{\GG_1}{\sigma}$ up to the renaming of the elements by $\tau$. Moreover, $\dfield{\GG_2}{\sigma}$ is a difference ring extension of $\dfield{\tau(\GG_1)}{\sigma}$. In a nutshell, $\dfield{\GG_1}{\sigma}$ is contained in $\dfield{\GG_2}{\sigma}$ by means of the embedding $\tau$. 

Consider our difference field $\dfield{\KK(n)}{\sigma}$ with $\sigma(n)=n+1$. Now define the evaluation function $\fct{\ev}{\KK(n)\times\NN}{\KK(n)}$ as follows.
For $\frac{p}{q}\in\KK(n)$ with $p,q\in\KK(n)$ and $\gcd(p,q)=1$,
\begin{equation}\label{Equ:EvRat}
\ev(\tfrac{p}{q},k)=
\begin{cases}
\frac{p(k)}{q(k)}&\text{if $q(k)\neq0$}\\
0&\text{if $q(k)=0$ \quad(pole case)};
\end{cases}
\end{equation}
here $p(k),q(k)$ with $k\in\NN$ denotes the evaluation of the polynomials at $n=k$.
Finally, we define the map $\fct{\tau}{\KK(n)}{\KK^{\NN}}$ by 
\begin{equation}\label{Equ:DefineTau}
\tau(f)=\langle \ev(f,k)\rangle_{k\geq0}=\langle \ev(f,0),\ev(f,1),\ev(f,2),\dots\rangle.
\end{equation}
Then one can easily see that $\tau$ is a difference ring homomorphism from $\dfield{\KK(n)}{\sigma}$ to $\dfield{\seqK}{\shiftS}$. In particular, $\tau$ is injective. Namely, take $f\in\KK(n)$ with $\tau(f)=\vect{0}$ i.e., $\ev(f,k)=0$ for all $k\in\NN_0$. Write $f=\frac{p}{q}$ with $p\in\KK[x]$ and $q\in\KK[x]\setminus\{0\}$. Since $q$ has only finitely many roots, $\tau(f)=\vect{0}$ implies that $p$ has infinitely many roots. Thus $p=0$ and therefore $f=0$. In summary, $\tau$ is a difference ring embedding. In particular, this implies that
$\dfield{\tau(\KK(n))}{\shiftS}$ forms a difference field also called the field of rational sequences.

Now consider our difference ring $\dfield{\KK(n)[x]}{\sigma}$ with $\sigma(x)=-x$ and $x^2=1$. We extend $\ev$ from $\KK(n)$ to $\KK(n)[x]$ as follows. For $f=f_0+f_1\,x$ with $f_0,f_1\in\KK(n)$ we define
$$\ev(f,k)=\ev(f_0,k)+(-1)^k\,\ev(f_1,k).$$
Furthermore we extend $\tau$ to $\KK(n)[x]$ with~\eqref{Equ:DefineTau} by our extended map $\ev$. Again one can straightforwardly verify that $\tau$ is a difference ring homomorphism. Moreover, $\tau$ is injective: Let $f\in\QQ(n)[x]$ with $f=f_0+f_1\,x$ such that $\tau(f)=\vect{0}$. Then $0=\ev(f_0,2k)+\ev(f_1,2k)=\ev(f_0+f_1,2k)$ for almost all $k\geq0$. As above, we conclude that $f_0+f_1=0$. Similarly, we get $0=\ev(f_0,2k+1)-\ev(f_1,2k+1)=\ev(f_0-f_1,2k)$ which implies that $f_0-f_1=0$. Consequently, $f_0=f_1=0$.
Summarizing, $\fct{\tau}{\KK(n)[x]}{\seqK}$ is a difference ring embedding. In particular, $\dfield{\tau(\KK(n)[x]}{\shiftS}$ is a difference ring which is the same as $\dfield{\KK(n)[x]}{\sigma}$ up to the renaming of the elements by $\tau$. $\dfield{\tau(\KK(n)[x]}{\shiftS}$ is also called the ring of rational sequences adjoined with the alternating sequence.

Now let $\AH\in\{A_h,A_a,A_{c(M)}\}$ for a finite set $M\in\NN$.
We will extend successively $\ev$ and thus $\tau$
for $\dfield{\EE^{\AH}_{d-1}}{\sigma}\leq\dfield{\EE^{\AH}_d}{\sigma}$ with $d=1,2,\dots$ as follows. For $\EE^{\AH}_0=\KK(n)[x]$ we have constructed $\ev$ and $\tau$. Now suppose that we obtained already the evaluation function $\fct{\ev}{\EE^{\AH}_{d-1}}{\KK}$ which yields a difference ring homomorphism $\fct{\tau}{\EE^{\AH}_{d-1}}{\seqK}$. 
In addition, assume that~\eqref{Equ:PolyRingForD} holds. Then take the sum representants~\eqref{Equ:Representants} such that $t^{(d)}_i=a^{(d)}_i+I_d$. 
Let 
$$f=\sum_{(e_1,\dots,e_{m_d})\in\NN_0^{m_d}} f_{(e_1,\dots,e_{m_d})}\,{t^{(d)}_1}^{e_1}\dots {t^{(d)}_{m_d}}^{e_{m_d}}\in\EE^{\AH}_d$$
where only finitely many $f_{(e_1,\dots,e_{m_d})}\in\EE^{\AH}_{d-1}$ are non-zero. Then we define
$$\ev(f,k)=\sum_{(e_1,\dots,e_{m_d})\in\NN_0^{m_d}} \ev(f_{(e_1,\dots,e_{m_d})},k)\,\expr({a^{(d)}_1})(k)^{e_1}\dots \expr({a^{(d)}_{m_d}})(k)^{e_{m_d}},$$
i.e., the sum variables are replaced by the corresponding \AH-sums depending on the variable $k$ and the expression from $\KK(n)[x]$ are treated by $\ev$ as described above. Finally, consider the function $\fct{\tau}{\EE^{\AH}_d}{\seqK}$ as given in~\eqref{Equ:DefineTau} with the extended $\ev$.
Again one can verify by simple calculations that $\tau$ forms a difference ring homomorphism; for further details see~\cite[Lemma~4.4]{DR1}.

What remains to show is that $\tau$ is also injective. Here we utilize the following result\footnote{If $\EE^{\AH}$ is free of $x$, this statement has been proven in~\cite{AlgebraicDF}.} that is implied immediately by Theorem~2.3 and Lemma~5.8 from~\cite{Schneider:17}.

\begin{theorem}\label{Thm:SigmaIsEmbedding}
Let $\dfield{\GG}{\sigma}$ with $\GG=\KK(n)[x][s_1]\dots[s_e]$ be a \sigmaSE-extension of the difference ring $\dfield{\KK(n)[x]}{\sigma}$. If $\fct{\tau}{\GG}{\seqK}$ is a difference ring homomorphism, then $\tau$ is injective.
\end{theorem}
\begin{proof}
Alternatively to the proof in~\cite{Schneider:17}, one can derive this result as follows: By Cor.~4.3 and Lemma~4.4 of~\cite{DR1} $\dfield{\GG}{\sigma}$ has no nilpotent elements. Hence  by~\cite[Cor.~1.24]{Singer:97} $\dfield{\GG}{\sigma}$ is a simple difference ring. I.e., any difference ideal of $\dfield{\GG}{\sigma}$ is either $\{0\}$ or $\GG$. Now consider $I=\{f\in\GG|\tau(f)=0\}$. Then this forms a difference ideal: if $h\in I$ then $\tau(\sigma(h))=\shiftS(\tau(h))=\shiftS(\vect{0})=\vect{0}$ and thus $\sigma(h)\in I$. Since $1\notin I$, $I\neq\GG$. Hence $I=\{0\}$, i.e., $\tau$ is injective. 
\end{proof}

Now we are ready to prove our main result.

\begin{theorem}\label{Thm:MainResult}
Let $\AH\in\{A_h,A_a,A_{c(M)}\}$. Then for any $d\in\NN$ we have that $\dfield{\EE^{\AH}_d}{\sigma}$ with $\EE^{\AH}_d=\KK(n)[x][t^{(1)}_1,\dots,t^{(1)}_{m_1}]\dots[t^{(d)}_1,\dots,t^{(d)}_{m_d}]$ is a \sigmaSE-extension of $\dfield{\KK(n)[x]}{\sigma}$. In particular, the map $\fct{\tau}{\EE^{\AH}_d}{\seqK}$ constructed above is a difference ring embedding.
\end{theorem}

\begin{proof}
Let $P\subset\ZZ[n]$ be the shifted denominators of the alphabet $\AH$. By definition $P$ is finite and we have that $P\subseteq\{n+1,2n+3,3n+4,3n+5,4n+5,4n+7,\dots\}$. Clearly, for any 
$a,b\in P$ we have that
$\disp(a,b)=\disp(b,a)=0$.
Now suppose that the theorem holds for $d-1$. I.e., $\dfield{\EE^{\AH}_{d-1}}{\sigma}$ is a \sigmaSE-extension of $\dfield{\KK(n)[x]}{\sigma}$. Note that this is a $P$-normalized \sigmaSE-extension.
Consider the difference ring extension
$\dfield{\EE^{\AH}_d}{\sigma}$ of $\dfield{\EE^{\AH}_{d-1}}{\sigma}$. By the construction from above there is a difference ring homomorphism $\fct{\tau}{\EE^{\AH}_d}{\seqK}$. In particular, by $\tau':=\tau|_{\EE^{\AH}_{d-1}}$ is a difference ring homomorphism from $\EE^{\AH}_{d-1}$ to $\seqK$ and thus $\tau'$ is injective by Theorem~\ref{Thm:SigmaIsEmbedding}.
Suppose that $\dfield{\HH}{\sigma}$ with $\HH=\EE^{\AH}_{d-1}[t^{(d)}_1,\dots,t^{(d)}_{r-1}]$ is a \sigmaSE-extension of $\dfield{\EE^{\AH}_{d-1}}{\sigma}$. Then $\dfield{\HH}{\sigma}$ is a $P$-normalized \sigmaSE-extension of $\dfield{\KK(n)[x]}{\sigma}$. Note that $\const{\HH}{\sigma}=\const{\KK(n)[x]}{\sigma}=\KK$ is a field. Now suppose that $\dfield{\HH[t^{(d)}_r]}{\sigma}$ is not a \sigmaSE-extension of $\dfield{\HH}{\sigma}$. Take $\beta:=\sigma(t^{(d)}_{r})-t^{(d)}_{r}\in\KK[n]_P[t^{(1)}_1,\dots,t^{(d-1)}_{m_{d-1}}]$. 
By Theorem~\ref{Thm:NotSigmaExt} it follows that there is a $g\in\KK(n)[x][t^{(1)}_1,\dots,t^{(d-1)}_{m_{d-1}}][t^{(d)}_1,\dots,t^{(d)}_{r-1}]$ 
such that~\eqref{Equ:Tele} holds. 
By Proposition~\ref{Prop:QRelation} we conclude that  $g\in\KK[t^{(1)}_1,\dots,t^{(d-1)}_{m_{d-1}}][t^{(d)}_1,\dots,t^{(d)}_{r-1}]$.  
Define
\begin{equation}\label{Equ:Definep}
p:=t^{(d)}_{r}-g\in\KK[t^{(1)}_1,\dots,t^{(d-1)}_{m_{d-1}}][t^{(d)}_1,\dots,t^{(d)}_{r}]\setminus\{0\}
\end{equation}
and consider
$q:=\sigma(p)-p$. By Proposition~\ref{Prop:TeleStructure} we conclude that $q\in\KK(n)[x][t^{(1)}_1,\dots,t^{(d-1)}_{m_{d-1}}]$.
Let $s(n)$ be the \AH-sum that is represented by $a^{(d)}_{r}$. This means that
$s(n)=\sum_{k=1}^n F(k-1)$ with $\expr(\beta)(n)=F(n)$. Let $G(n)=\expr(g)(n)$. Note that $F(k)$ has no pole for all $n$ with $n\geq0$ by definition of our \AH-sums. Furthermore $g$ is free of $n$ and thus also $G(n)$ is free of rational expression in $n$ (only the arising \AH-sums depend on $n$ in the outermost summation sign). Hence $G(n)$ and $G(n+1)$ have no poles for all $n$ with $n\geq0$. Therefore we get
\begin{align*}
G(k+1)-G(k)=&\expr(g)(k+1)-\expr(g)(k)=\expr(\sigma(g))(k)-\expr(g)(k)\\
&=\expr(\sigma(g)-g)(k)=\expr(\beta)(k)=F(k)
\end{align*}
for all $k\geq0$. Hence $F(k)=G(k+1)-G(k)$
for all $k\geq0$ and summing this equation over $k$ from $0$ to $n-1$ we conclude by telescoping that for all $n\geq0$ we have that
$$s(n)=\sum_{k=1}^nF(k-1)=\sum_{k=0}^{n-1}F(k)=\sum_{k=0}^{n-1}\big(G(k+1)-G(k)\big)=G(n)-c$$
with $c:=G(0)\in\KK$. Consequently, 
\begin{align*}
\tau(p)&=\tau((a^{(d)}_{r}+I)-g)=\tau(a^{(d)}_{r})-\tau(g)=\langle s(\lambda)\rangle_{\lambda\geq0}-\langle G(\lambda)\rangle_{\lambda\geq0}\\
&=\langle s(\lambda)-G(\lambda)\rangle_{\lambda\geq0}=\langle c,c,c,\dots\rangle
\end{align*}
and therefore
$$\tau'(q)=\tau(q)=\tau(\sigma(p)-p)=\sigma(\tau(p))-\tau(p)=\vect{0}.$$
Since $\tau'$ is injective, $q=0$.
Hence Lemma~\ref{Lemma:LBarSigmaCommute} and the linearity of $\bar{L}$ yield
$$\sigma(\bar{L}(p+I))-\bar{L}(p+I)=\bar{L}(\sigma(p+I))-\bar{L}(p+I)=\bar{L}(\sigma(p+I)-(p+I))=\bar{L}(q+I)=\bar{L}(0+I)=0.$$
Therefore $\sigma(u)=u$ with $u:=\bar{L}(p+I)$. By~\eqref{Equ:Definep} it follows that $u\in\bar{\VV}$ by Lemma~\ref{Lemma:LinearConstant}.
Write $p=p'+I_d\in\EE^{\AH}_{d}\setminus\{0\}$ with $p'\in\AA^{\AH}_d$. Note that $L(p')\neq0$ since $p\neq0$ (\ie $p\neq0+I_d$).
Thus $u=\bar{L}(p+I)=L(p')\neq0$, a contradiction to Lemma~\ref{Lemma:LinearConstant}.
\end{proof}

\begin{example}
By Theorem~\ref{Thm:MainResult}
the difference ring $\dfield{\EE^{\AH}_4}{\sigma}$, or equivalently the difference ring $\dfield{\bar{\EE}^{\AH}_4}{\bar{\sigma}}$, of Example~\ref{Exp:LinearAlgebra} with $\AH=A_h$ is a \sigmaSE-extension of $\dfield{\KK(n)[x]}{\sigma}$. In particular, we obtain the difference ring embedding
$\fct{\tau}{\bar{\EE}^{\AH}_4}{\seqK}$ given by
\begin{align*}
\tau(\Sv_{1})=&\langle S_1(\lambda)\rangle_{\lambda\geq0},&
\tau(\Sv_{2})=&\langle S_{2}(\lambda)\rangle_{\lambda\geq0},\\
\tau(\Sv_{2,1})=&\langle S_{2,1}(\lambda)\rangle_{\lambda\geq0},&
\tau(\Sv_{3})=&\langle S_{3}(\lambda)\rangle_{\lambda\geq0},\\
\tau(\Sv_{2,1,1})=&\langle S_{2,1,1}(\lambda)\rangle_{\lambda\geq0},&
\tau(\Sv_{3,1})=&\langle S_{3,1}(\lambda)\rangle_{\lambda\geq0},\\
\tau(\Sv_{4})=&\langle S_{4}(\lambda)\rangle_{\lambda\geq0}.
\end{align*}
Therefore we get the polynomial ring
$$\tau(\EE^{\AH}_4)=\tau(\KK(n)[x])[\tau(\Sv_{1})][\tau(\Sv_{2})][\tau(\Sv_{2,1}),\tau(\Sv_{3})][\tau(\Sv_{2,1,1}),\tau(\Sv_{3,1}),\tau(\Sv_4)],$$
i.e., the sequences generated by the basis sums
$S_{1}(n)$, $S_{2}(n)$, $S_{2,1}(n)$, $S_{3}(n)$, $S_{2,1,1}(n)$, $S_{3,1}(n)$, and $S_4(n)$ induced by the quasi-shuffle algebra are algebraically independent over the rational sequences adjoined with the alternating sequence. Obviously, these sums are also algebraically independent over the field of rational sequences $\tau(\KK(n))$. 
\end{example}

By induction it follows by Theorem~\ref{Thm:MainResult} that there is a difference ring embedding $\fct{\tau}{\EE^{\AH}}{\seqK}$. In particular, following the construction from above, it has the form as stated in the following corollary.

\begin{corollary}
Let $\AH\in\{A_h,A_a,A_{c(M)}\}$. Then $\const{\EE^{\AH}}{\sigma}=\KK$ and there is a difference ring embedding $\fct{\tau}{\EE^{\AH}}{\seqK}$. Furthermore,
$$\tau(\EE^{\AH})=\tau(\KK(n)[x])[\tau(t^{(1)}_1),\dots,\tau(t^{(1)}_{m_1})][\tau(t^{(2)}_1),\dots,\tau(t^{(2)}_{m_2})]\dots$$
forms a polynomial ring with the sequences
$$\tau(t^{(d)}_j)=\langle\expr(a^{(d)}_j)(\lambda)\rangle_{\lambda\geq0}$$
where $\expr(a^{(d)}_1),\dots,\expr(a^{(d)}_{m_d})$ are the \AH-sums of weight $d$ which  cannot be reduced further by the quasi-shuffle algebra.
\end{corollary}

In the following we define $M_k=\{1,2,\dots,k\}$ and consider the general case  $\AH=A_c$.
Let $\dfield{\EE_d^{A_c(M_k)}}{\sigma_{d,k}}$ for $k\in\NN$ and $d\in\NN_0$ be the reduced difference ring for the cyclotomic harmonic sums $\AH=A_c(M_k)$ with weight $\leq d$ as constructed in Section~\ref{Sec:ReduceDR}. By construction we get the chain of difference ring extensions
\begin{equation}\label{Equ:ChainDepth}
\dfield{\EE^{A_{c(M_k)}}_0}{\sigma_{0,k}}\leq\dfield{\EE^{A_{c(M_k)}}_1}{\sigma_{1,k}}\leq \dfield{\EE^{A_{c(M_k)}}_{2}}{\sigma_{2,k}}\leq\dots
\end{equation}
where the cyclotomic alphabet remains unchanged and the weights of the extensions are increased step-wise.
In particular, these extensions form \sigmaSE-extensions by Theorem~\ref{Thm:MainResult}.
Furthermore, by Remark~\ref{Remark:IncreaseAlphabet}
we get the chain of difference ring extensions
\begin{equation}\label{Equ:ChainAlphabet}
\dfield{\EE^{A_{c(M_1)}}_d}{\sigma_{d,1}}\leq\dfield{\EE^{A_{c(M_2)}}_d}{\sigma_{d,2}}\leq\dfield{\EE^{A_{c(M_3)}}_d}{\sigma_{d,3}}\leq\dots
\end{equation}
where the weights do not increase but the cyclotomic alphabets are increased step-wise.
Again, these extensions form \sigmaSE-extensions by Theorem~\ref{Thm:MainResult}.

By~\eqref{Equ:ChainDepth}
 and~\eqref{Equ:ChainAlphabet} we get the \sigmaSE-extensions
\begin{equation}\label{ChainDepthAlpha}
 \dfield{\EE^{A_{c(M_k)}}_k}{\sigma_{k,k}}\leq\dfield{\EE^{A_{c(M_{k+1})}}_k}{\sigma_{k,k+1}}\leq \dfield{\EE^{A_{c(M_{k+1})}}_{k+1}}{\sigma_{k+1,k+1}}
\end{equation}
for all $k\in\NN$. Here the cyclotomic alphabets and weights are increased simultaneously step by step.

Finally, define,
$$\EE^{A_c}:=\bigcup_{k\in\NN_0}\EE^{A_{c(M_k)}}_k.$$
Since $\EE^{A_{c(M_k)}}_k$ is a subring of $\EE^{A_{c(M_r)}}_r$ for any $r,k\in\NN$ with $k\leq r$, $\EE^{A_c}$ forms a ring. 
Furthermore define the function
$\fct{\sigma}{\EE^{A_c}}{\EE^{A_c}}$ as follows. For $f\in\EE^{A_c}$ take $k\in\NN$ minimal such that $f\in\EE^{A_{c(M_k)}}_k$. Then we define $\sigma(f):=\sigma_{k,k}(f)$. 

\begin{lemma}\label{Lemma:EAcDF}
$\dfield{\EE^{A_c}}{\sigma}$ is a difference ring which is a difference ring extension of $\dfield{\EE^{A_c(M_r)}_r}{\sigma_{r,r}}$.
\end{lemma}
\begin{proof}
Let $f\in\EE^{A_{c(M_r)}}_r$ be arbitrary but fixed. Take $l\in\NN$ being minimal such that $f\in\EE^{A_{c(M_l)}}_l$. Then $\sigma(f)=\sigma_{l,l}(f)$. By iterative application of~\eqref{ChainDepthAlpha} it follows that $\sigma_{r,r}(f)=\sigma_{l,l}(f)=\sigma(f)$. Therefore 
\begin{equation}\label{Equ:SigmaRestricted}
\sigma(f)=\sigma_{r,r}(f)\text{ for all }f\in\EE^{A_{c(M_r)}}_r. 
\end{equation}
Let $f,g\in\EE^{A_c}$.
Now take $r$ such that $f,g,\sigma(f)\in\EE^{A_{c(M_r)}}_r$. Then by~\eqref{Equ:SigmaRestricted} and the fact that $\sigma_{r,r}$ is a difference ring automorphism, we conclude that $\sigma(f\,g)=\sigma(f)\,\sigma(g)$ and $\sigma(f+g)=\sigma(f)+\sigma(g)$. Similarly, one can show that $\sigma$ is bijective. Hence $\sigma$ is a difference ring automorphism, i.e., $\dfield{\EE^{A_c}}{\sigma}$ is a difference ring. Finally, by~\eqref{Equ:SigmaRestricted} we conclude that $\dfield{\EE^{A_c}}{\sigma}$ is a difference ring extension of $\dfield{\EE^{A_c(M_r)}_r}{\sigma_{r,r}}$ for any $r\in\NN$.
\end{proof}

At the end, with Lemma~\ref{Lemma:EAcDF} and~\eqref{ChainDepthAlpha} we obtain the following chain of difference ring extensions
$$\dfield{\EE^{A_c(M_1)}_1}{\sigma_{1,1}}\leq\dfield{\EE^{A_c(M_2)}_2}{\sigma_{2,2}}\leq\dfield{\EE^{A_c(M_3)}_3}{\sigma_{3,3}}\leq\dots\leq\dfield{\EE^{A_c}}{\sigma}$$
where each $\dfield{\EE^{A_c(M_{k-1})}_{k-1}}{\sigma_{k-1,k-1}}\leq\dfield{\EE^{A_c(M_k)}_k}{\sigma_{k,k}}$ forms a nested \sigmaSE-extension. 

Summarizing, we obtain the following result.

\begin{corollary}
Let $\AH=A_c$. Then $\const{\EE^{\AH}}{\sigma}=\KK$ and there is a difference ring embedding $\fct{\tau}{\EE^{\AH}}{\seqK}$. In particular,
\begin{equation}\label{Equ:GeneratorsForEc}
\tau(\EE^{\AH})=\tau(\KK(n)[x])[\tau(s^{(1)}_1),\dots,\tau(s^{(1)}_{\mu_1})][\tau(s^{(2)}_1),\dots,\tau(s^{(2)}_{\mu_2})]\dots
\end{equation}
forms a polynomial ring with the sequences
$$\tau(s^{(k)}_j)=\langle\expr(b^{(k)}_j)(n)\rangle_{n\geq0}$$
where $\expr(b^{(k)}_1),\dots,\expr(b^{(k)}_{\mu_{k}})$ are all $A_{c(M(k))}$-sums of weight $\leq k$ which are not contained in $A_{c(M(k-1))}$ and which cannot be reduced further by the quasi-shuffle algebra.
\end{corollary}

We remark that the construction of the difference ring $\dfield{\EE^{A_c}}{\sigma}$ can be accomplished in many different ways yielding a different sorting of the sum-generators in~\eqref{Equ:GeneratorsForEc}. Here we chose a construction where each extension step is built only by finitely many \sigmaSE-extensions.

\section{Conclusion}\label{Sec:Conclusion}

We showed that the reduced representation (induced by the quasi-shuffle algebra) of the harmonic sums with the alphabet $A_h$ and $A_a$ and the cyclotomic harmonic sums with the alphabet $A_c$ constitute a tower of \sigmaSE-extensions. This means we can model these nested sums in a difference ring extension where the constants are not extended. As a consequence these nested sums can be embedded into the ring of the sequences. Furthermore, this shows that the found relations due to the quasi-shuffle algebra are complete, i.e., there are no further relations that might occur between the sequences of the nested sums. We emphasize that these techniques are not restricted to cyclotomic harmonic sums. Most of the ideas presented in this article can be carried over to the more general class of \rpisiSE-extensions~\cite{DR1,DR2} (which contains the class of \sigmaSE-extensions). Since this class covers, e.g.,  \hbox{$q$--}hypergeometric sequences, one can also treat the generalized (cyclotomic) harmonic sums 
with the alphabet
$$A_g:=A\cap(\NN\times\NN\times\NN\times S)$$
for certain finite sets $S\subset\KK$,
the inverse binomial sums and also their $q$-versions in this setting. We expect that the results and tools presented in this work will be helpful to show algebraic independence in the ring of sequences also for these more general classes of nested sums.


\end{document}